\documentclass[11pt,reqno]{article}
\usepackage[margin=1in]{geometry}
\usepackage{xspace}
\usepackage{xcolor}
\usepackage{amsmath,amsfonts,amsthm,amssymb}
\usepackage{array}
\usepackage{mathtools}
\usepackage{verbatim}
\usepackage{mathrsfs} 
\usepackage{hyperref} 
\usepackage[toc,page]{appendix}
\usepackage{cleveref}
\usepackage{graphicx}
\usepackage{algpseudocode}
\usepackage{tikz}
\usetikzlibrary{decorations.pathreplacing}

\usepackage{hyperref} 
\usepackage[ruled,vlined]{algorithm2e}
\usepackage{titlesec}
\titleformat{\subsubsection}[runin]
  {\normalfont\normalsize\bfseries}
  {\thesubsubsection}
  {1em}
  {}
\newtheorem{theorem}{Theorem}[section]
\newtheorem{corollary}[theorem]{Corollary}
\newtheorem{lemma}[theorem]{Lemma}
\newtheorem{proposition}[theorem]{Proposition}
\theoremstyle{definition}
\newtheorem{definition}[theorem]{Definition}

\newtheorem{claim}[theorem]{Claim}

\newcommand{\poly}{\mathsf{poly}}
\newcommand{\polylog}{\mathrm{polylog}}
\newcommand{\quasipoly}{\mathrm{quasipoly}}
\newcommand{\eps}{\varepsilon}

\title{Improved Explicit Near-Optimal Codes in the High-Noise Regimes}
 \author{Xin Li\thanks{Department of Computer Science,  Johns Hopkins University, Email: lixints@cs.jhu.edu. Supported by NSF CAREER Award CCF-1845349 and NSF Award CCF-2127575.}
 \and Songtao Mao\thanks{Department of Computer Science, Johns Hopkins University, Email: smao13@jhu.edu. Supported by NSF Award CCF-2127575.}
 }
\date{}

\begin{document}
\maketitle

\begin{abstract}
We study uniquely decodable codes and list decodable codes in the high-noise regime, specifically codes that are uniquely decodable from  $\frac{1-\varepsilon}{2}$ fraction of errors and list decodable from $1-\varepsilon$ fraction of errors. We present several improved explicit constructions that achieve near-optimal rates, as well as efficient or even linear-time decoding algorithms. Our contributions are as follows.

\begin{itemize}
    \item \textbf{Explicit Near-Optimal Linear Time Uniquely Decodable Codes:} We construct a family of explicit $\mathbb{F}_2$-linear codes with rate $\Omega(\varepsilon)$ and alphabet size $2^{\poly \log(1/\varepsilon)}$, that are capable of correcting $e$ errors and $s$ erasures whenever $2e + s < (1 - \varepsilon)n$ in linear-time. To the best of our knowledge, this is the first fully explicit linear time decodable code over an alphabet of size $2^{o(1/\varepsilon)}$, that beats the $O(\varepsilon^2)$ rate barrier.
    \item \textbf{Explicit Near-Optimal List Decodable Codes:} We construct a family of explicit list decodable codes with rate $\Omega(\varepsilon)$ and alphabet size $2^{\poly \log(1/\varepsilon)}$, that are capable of list decoding from $1-\varepsilon$ fraction of errors with a list size $L = \exp\exp\exp(\log^{\ast}n)$ in polynomial time. To the best of our knowledge, this is the first fully explicit list decodable code with polynomial-time list decoding over an alphabet of size $2^{o(1/\varepsilon)}$, that beats the $O(\varepsilon^2)$ rate barrier.
    \item \textbf{List Decodable Code with Near-Optimal List Size:} We construct a family of explicit list decodable codes with an optimal list size of $O(1/\varepsilon)$, albeit with a suboptimal rate of $O(\varepsilon^2)$, capable of list decoding from $1-\varepsilon$ fraction of errors in polynomial time. Furthermore, we introduce a new combinatorial object called \textit{multi-set disperser}, and use it to give a family of list decodable codes with near-optimal rate $\frac{\varepsilon}{\log^2(1/\varepsilon)}$ and list size $\frac{\log^2(1/\varepsilon)}{\varepsilon}$, that can be constructed in probabilistic polynomial time and decoded in deterministic polynomial time.
\end{itemize}
Our techniques are based on plurality analysis and graph-concatenated codes, which are widely used in the literature. We also introduce new decoding algorithms that may prove valuable for other graph-based codes.

\end{abstract}
\newpage
\pagenumbering{arabic}

\section{Introduction}\label{sec_intr}
Error-correcting codes are fundamental objects designed to ensure the accurate transmission of data across channels subject to noise or adversarial errors.\ They can be described simply as a function $ \mathcal{C}: \tilde{\Sigma}^k \mapsto \Sigma^n$, which maps a $k$-symbol message over one alphabet $\tilde{\Sigma}$ to an $n$-symbol codeword over another alphabet $\Sigma$ (in many cases we simply use $\Sigma=\tilde{\Sigma}$). Given any code, the two most important parameters are the information rate $R$ and the distance $d$. The rate $R$ is defined as $ R = \frac{k\log(|\tilde{\Sigma}|)}{n\log(|\Sigma|)}$, which measures the amount of information in any codeword and hence represents the efficiency of the code. A higher rate is preferred, as it reduces the redundancy in the encoded message. The  distance $d$, is defined as the smallest Hamming distance between any two distinct codewords. The ratio of the distance to the codeword length is referred to as the relative distance $\delta$. These two parameters are important because in many situations, they characterize exactly the number (or fraction) of adversarial errors that the code can correct. For example, it is well known that a code with distance $d$ can correct exactly up to $\lfloor \frac{d-1}{2} \rfloor$ adversarial errors (symbol corruptions), if one wishes to recover the original message uniquely. This corresponds to the well studied area of \emph{unique decoding}. On the other hand, with a slight relaxation of outputting a small list of possible messages (that contains the correct message), one can hope to tolerate close to $d$ adversarial errors. This corresponds to another well studied area of \emph{list decoding}. Therefore, a larger distance is also preferred.

However, it is also well known that $R$ and $d$ (or $\delta$) cannot both be large, and in fact there are many well established trade-offs between these two parameters. For example, one of the most general bounds, known as the Singleton bound, states that any code must satisfy $\delta + R \leq 1$. If one restricts the alphabet size, then tighter bounds (such as the Hamming bound) are known. One of the major goals in coding theory is to design explicit codes with good trade-offs between $R$, $\delta$, and the alphabet size. Equally importantly, it is desirable to have efficient or fast decoding algorithms for the explicitly designed codes. Indeed, most of the research in algorithmic coding theory focuses on the above two goals, and so does this paper. Here, we focus on the case of high error, that is, to uniquely decode from $\frac{1-\eps}{2}$ fraction or list decode from $1-\eps$ fraction of adversarial errors, for any constant $\eps>0$. For unique decoding, the goal is to try to achieve the smallest possible alphabet size, while for list decoding the goal is to try to achieve the smallest possible alphabet size as well as the smallest possible list size. Simultaneously, we will also try to design fast decoding algorithms. We now discuss previous works in these two cases in more details below. 

\subsubsection*{Unique Decoding} In the regime of unique decoding, by the Singleton bound one cannot hope to correct more than $1/2$ fraction of errors. Therefore, the natural goal is to construct codes that can correct up to $\frac{1-\eps}{2}$ fraction of errors, for any constant $\eps>0$, which corresponds to a relative distance of roughly $\delta=1-\eps$. For such codes, the Hamming bound and the Plotkin bound imply that the alphabet size has to be $\Omega(1/\eps)$ if one wishes to have a positive rate, while the Gilbert-Varshamov bound implies that an alphabet of size $O(1/\eps)$ is enough to achieve rate $\Omega(\eps)$. 

In terms of explicit codes, Reed-Solomon codes \cite{reed1960polynomial} achieve rate $\eps$  with an alphabet size of $\Omega(n)$, which is necessary for any code that meets the Singleton bound. On the other hand, explicit algebraic geometry codes can achieve rate $\Omega(\eps)$ with an alphabet size $\poly(1/\eps)$, however the decoding algorithm for such codes runs in a fairly large polynomial time.

Over the years (near) linear time decodable/encodable codes for $\frac{1-\eps}{2}$ fraction of errors have also been studied, and most of the constructions are based on variants of expander-based codes \cite{sipser1996expander}. Such fast decoding algorithms enable the codes to be used with a lot of practical benefits, and hence is the focus of much research. For example, \cite{guruswami2001expander} introduced a family of codes capable of correcting $\frac{1-\eps}{2}$ fraction of errors, with rate $\Omega(\varepsilon)$ and slightly super linear time decoding,  or rate $\Omega(\varepsilon^2)$ and linear-time decoding. Subsequently, building on the near-MDS linear-time codes for decoding from erasures introduced in \cite{alon1995linear}, \cite{guruswami2002near} constructed codes that are both linear-time encodable and decodable, that can correct up to a $1-R-\varepsilon$ fraction of errors with rate $R$. However, the alphabet size in all these codes is exponential in $1/\varepsilon$. In a subsequent improvement \cite{rom2006improving}, the authors constructed a family of codes with rate $ \Omega(\varepsilon) $ and alphabet size  $2^{2^{\poly \log \log (1/\eps)}}$. However, the unique decoding algorithm still runs in slightly super-linear time. We refer the reader to Table~\ref{tab_uni} for a detailed list of previous works.

\subsubsection*{List Decoding}
As stated before, list decoding is another important and well studied area in coding theory, which applies when the fraction of errors exceeds $\frac{\delta}{2}$, rendering unique decoding infeasible. In such situations, the decoder is allowed to output a small list with $L \geq 1$ codewords. Formally, a code $\mathcal{C} \subseteq \Sigma^n$ is called $(\rho, L)$-(combinatorially) list decodable if, for every $y \in \Sigma^n$, there are at most $L$ codewords in $\mathcal{C}$ whose relative distance from $y$ is less than $\rho$. Beyond its natural application, list decoding is also closely rated to many other areas in theoretical computer science.

A fundamental goal in the study of list decodable codes is to construct codes that achieve the \emph{list decoding capacity}, which can be summarized as follows. Fix any $q \geq 2$, $0 \leq \rho \leq 1-1/q$, and $\eps > 0$, then there exist $(\rho, L)$-list decodable codes over an alphabet of size $q$ with rate $R \leq 1-H_q(\rho)-\eps$ and $L=O(1/\eps)$. On the other hand, for any $(\rho, L)$ code with rate $1-H_q(\rho)+\varepsilon$, we must have $L=q^{\Omega(n)}$ where $n$ is the codeword length. In particular, when the alphabet size is sufficiently large, a random code with rate $R$ will, with high probability, be list decodable from $1-R-\varepsilon$ fraction of errors. This gets close to the Singleton bound. However, providing explicit constructions with the best possible trade-off between these parameters and efficient decoding algorithms remains a significant challenge.

Similar to the case of unique decoding, here we consider the high-error regime, and study codes that can list decode from $1-\eps$ fraction of errors. In this case, using the probabilistic method it is easy to show the existence of $(1-\eps, O(1/\eps))$-list decodable codes with rate $\Omega(\eps)$ and alphabet size $O(1/\eps^2)$. However, for a while only explicit constructions with rate $\Omega(\eps^2)$ are known, until Guruswami \cite{guruswami2004better} gave the first explicit construction that can achieve rate $\Omega(\eps/\log(1/\eps))$ assuming one has an explicit construction of optimal \emph{strong seeded extractor}. The alphabet size is $2^{O(\eps^{-1}\log (1/\eps))}$ and the list size is $O(1/\eps)$. However, to date no explicit construction of such optimal strong seeded extractors is known\footnote{They can be constructed in probabilistic polynomial time though.}; and with known explicit constructions of strong seeded extractors, the code in \cite{guruswami2004better} only achieves list size $2^{O(\sqrt{n \log n})}$, where $n$ is the codeword length. In the following improvement \cite{rom2006improving}, the authors again reduced the alphabet size to $2^{2^{\poly \log \log (1/\eps)}}$ while keeping the other parameters roughly unchanged. We also mention that through another line of research which constructs capacity achieving list decodable codes \cite{guruswami2008explicit,guruswami2009artin,guruswami2013linear,kopparty2015list,guruswami2013list,guruswami2016explicit,hemenway2019local,kopparty2020list,guo2021efficient,guruswami2022optimal}, we now have explicit  codes with rate $R$ that can list decode from $1-R-\varepsilon$ fraction of errors. By setting for example $R=\eps$, this also gives explicit codes that can list decode from $1-\varepsilon$ fraction of errors with rate $\Omega(\eps)$. In \cite{alrabiah2024ag}, the authors showed that a code with rate $R$, achieving a list decoding radius of $1 - R - \varepsilon$ and a constant list size, requires an alphabet size of $2^{\Omega(1/\varepsilon)}$, and the best known list size, as in \cite{guo2021efficient}, is still $2^{\poly(1/\varepsilon)}$. 

Another line of research in list decoding focuses on minimizing the list size while approaching the Singleton bound. A series of works \cite{guo2024improved,brakensiek2023generic,guo2023randomly,alrabiah2023randomly,brakensiek2024ag} demonstrated that most Reed-Solomon codes and algebraic-geometric codes are list-decodable with an optimal list size. Later, \cite{ron2024efficient} extended the definition of GM-MDS to 'polynomial ideal codes' and proposed an efficient list decoding algorithm for random polynomial ideal codes. However, these codes are not fully explicit, as they are obtained through random puncturing. For explicit constructions, progress has been made in reducing the list size. Several works, including \cite{dvir2012subspace,kopparty2023improved,tamo2024tighter,srivastava2024improved}, have achieved optimal list-decodable codes with a constant list size. Most recently, in \cite{chen2024explicit}, the authors presented explicit codes that achieve list decoding capacity with an optimal list size of $ L = \lceil \frac{1}{\varepsilon} \rceil $. Despite these advances, all such codes still require large alphabet sizes, which are polynomial in $n$ for (Folded) Reed-Solomon codes or exponential in $1/\varepsilon$ for AG codes.

We refer the reader to Table~\ref{tab_list} for a detailed list of previous works.

\subsection{Our Results}

We give explicit constructions of codes in both the unique decoding regime and the list decoding regime. In the former, our codes can correct up to $\frac{1-\varepsilon}{2}$ fraction of errors; and in the latter, our codes can list decode from up to $1-\varepsilon$ fraction of errors. Our codes significantly improve previous results in several aspects, which we discuss below.

In the unique decoding regime, our codes achieve rate $\Omega(\eps)$ and alphabet size $2^{\poly \log(1/\eps)}$, together with truly linear decoding time for any constant $\eps>0$. This improves upon the best known previous results of \cite{guruswami2002near}, which has linear time decoding and rate $\Omega(\eps)$, but with alphabet size $2^{\poly(1/\eps)}$; and the subsequent work \cite{rom2006improving}, which has rate $ \Omega(\varepsilon) $ and alphabet size  $2^{2^{\poly \log \log (1/\eps)}}$, but super-linear decoding time\footnote{As noted before, algebraic geometry codes can correct up to $\frac{1-\varepsilon}{2}$ fraction of errors with rate $\Omega(\eps)$ and alphabet size $\poly(1/\eps)$, but the decoding takes a large polynomial time.}. Specifically, we have the following theorem.

\begin{theorem}\label{thm:intro1}
    For any $\varepsilon > 0$, there exists an explicit family of $\mathbb{F}_2$-linear codes over an alphabet of size $\quasipoly\left(1/\varepsilon\right) = \exp(\polylog(1/\varepsilon))$, which have rate $\Omega(\varepsilon)$ such that a code with block length $n$ in the family can be decoded from $e$ errors and $s$ erasures as long as $2\cdot e + s < (1-\varepsilon)n$ in time $O(n \cdot\polylog(1/\varepsilon))$.
\end{theorem}

To the best of our knowledge, this is the first fully explicit code to achieve a linear-time decoding algorithm with an alphabet size smaller than $2^{O(1/\varepsilon)}$, and rate beating the $O(\varepsilon^2)$ barrier. Indeed, the best known previous linear-time uniquely decodable code with rate $\Omega(\eps)$ was presented in \cite{guruswami2002near}, with an alphabet size of $2^{\poly (1/\eps)}$. Our codes thus give a substantial improvement.

As in previous works \cite{guruswami2002near, rom2006improving}, our construction is based on the use of certain extremal graphs. Specifically, we use an explicit disperser with constant error and constant entropy loss. In more details, a bipartite graph $G : (L \sqcup R, E)$ with uniform left degree $D$ is called a $(K, \delta)$-disperser if for every subset of left vertices $S \subseteq L(G)$ of size at most $K$, the size of its neighbors $|\Gamma(S)|$ is at least $(1 - \delta)|R(G)|$. Here, $\delta$ is the error, and $\log\left(\frac{KD}{|R(G)|}\right)$ is called the entropy loss of this disperser. Optimally, but non explicitly, there exists a disperser with constant error and entropy loss, and left degree $D=\Theta(\log(\frac{1}{\varepsilon}))$ when $K=\Theta(N)$. Such a disperser can be constructed in probabilistic polynomial time. Using this disperser, we can further reduce the alphabet size to $\poly\left(1/\varepsilon\right)$. 

\begin{theorem}\label{thm:intro2}
    For any $\varepsilon > 0$, there exists a family of $\mathbb{F}_2$-linear codes over an alphabet of size $\poly(1/\varepsilon)$, which have rate of $\Omega(\varepsilon)$ such that a code with block length $n$ in the family can be decoded from $e$ errors and $s$ erasures as long as $2\cdot e + s < (1-\varepsilon)n$ in time $O(n \cdot \polylog(1/\varepsilon))$, and can be constructed in probabilistic time $\poly(n,\log(1/\varepsilon))$ with success probability at least $1-\exp(-\log(1/\varepsilon)\cdot n)$. 
\end{theorem}

We refer the reader to Table~\ref{tab_uni} for a detailed comparison of our results and previous results.

\begin{table}[h!]
    \centering
    \renewcommand{\arraystretch}{1.5} 
    \begin{tabular}{
    | >{\centering\arraybackslash}m{2cm} 
    | >{\centering\arraybackslash}m{2cm} 
    | >{\centering\arraybackslash}m{4cm} 
    | >{\centering\arraybackslash}m{2.2cm} 
    | >{\centering\arraybackslash}m{2.5cm}
    |
    }
    \hline
   & \textbf{Rate} $R$  & \textbf{Alphabet size} &  \textbf{Poly-time construction} & \textbf{Decoding time} \\
    \hline
    \cite{guruswami2001expander}  & $\varepsilon$ & $\exp(\frac{1}{\varepsilon})$ &  D &  $n^{1+\gamma}$ \\
    \hline
    \cite{guruswami2001expander}  & $\varepsilon^2$ & $\exp(\frac{1}{\varepsilon^2})$ &  D &  $n/\varepsilon^2$ \\
    \hline
    \cite{guruswami2002near}  & $\varepsilon$ & $\exp(\poly(\frac{1}{\varepsilon}))$ &  D &  $n/\varepsilon^6$ \\
    \hline
    \cite{rom2006improving} & $\varepsilon$ & $\exp(\exp(\poly\log\log(\frac{1}{\varepsilon})))$ &  D & $n^{1+\gamma}$ \\
    \hline
    \cite{rom2006improving} & $\varepsilon$ & $\exp(\poly\log(1/\varepsilon))$ &  P & $n^{1+\gamma}$ \\
    \hline
    Our Work  & $\varepsilon$ & $\exp(\poly\log(\frac{1}{\varepsilon}))$ &  D &  $n\cdot\poly\log(1/\varepsilon)$ \\
    \hline
    Our Work  & $\varepsilon$ & $\poly(1/\varepsilon)$ &  P &  $n\cdot\poly\log(1/\varepsilon)$ \\
    \hline
    \end{tabular}
    \caption{Uniquely decodable codes capable of decoding a $(1-\varepsilon)/2$ fraction of errors, where $\gamma$ is a constant affecting the rate by a constant factor. The fourth column indicates whether such codes can be obtained in polynomial time, either deterministically or probabilistically. We use D for deterministic polynomial time and P for probabilistic polynomial time. In order to keep things concise here, we have omitted the $O(\cdot)$ and $\Omega(\cdot)$ notations.}
    \label{tab_uni}
\end{table}

In the list decoding regime, we provide several improved constructions. Our first construction also achieves rate $\Omega(\eps)$ and alphabet size $2^{\poly \log(1/\eps)}$, with a list size of $\exp(\exp(\exp(\log^{\ast}n)))$. This significantly improves the previous work of \cite{rom2006improving}, which has rate $ \Omega(\varepsilon/\poly\log(1/\eps)) $, alphabet size  $2^{2^{\poly \log \log (1/\eps)}}$, and list size $2^{O(\sqrt{n \log n})}$. Compared to the work of \cite{guo2021efficient}, their list size is smaller ($2^{\poly(1/\varepsilon)}$) when $\eps$ is a constant, but our alphabet size is much smaller ($2^{\poly \log(1/\eps)}$ compared to $2^{O(1/\varepsilon^2\log(1/\varepsilon))}$ in \cite{guo2021efficient}). If we use an optimal (but non-explicit) disperser, then we can further reduce the alphabet size to $\poly(1/\varepsilon)$. Specifically, we have the following theorem.

\begin{theorem}\label{thm: intro3}
For any $\varepsilon > 0$, there exists an explicit family of codes over an alphabet of size $\quasipoly\left(1/\eps\right)$ which have rate $\Omega(\varepsilon)$ such that a code with block length $n$ in this family can be list decoded up to $(1-\varepsilon)$ fraction of errors with list size $L = \exp\exp\exp_{\eps}(\log^\ast n)$ in time $\poly_{\varepsilon}(n)$. Moreover, there exists a family of codes over an alphabet of size $\poly(1/\varepsilon)$ that maintains the same rate and list-decoding properties, and a code with block length $n$ in this family can be constructed in probabalistic time $\poly(n,\log(1/\varepsilon))$ with success probability at least $1-\exp(-\log(1/\varepsilon)\cdot n)$.
\end{theorem}

To the best of our knowledge, this is the first fully explicit list decodable code with polynomial-time list decoding that beats both the $2^{\Omega(1/\varepsilon)}$ alphabet size barrier and the $O(\varepsilon^2)$ rate barrier. Indeed, the constructions in \cite{guruswami2004better} and subsequently \cite{rom2006improving} need to use non-explicit optimal strong seeded extractor to achieve this, otherwise their list sizes become as large as $2^{O(\sqrt{n \log n})}$, and hence require super-polynomial decoding time. 

We also give constructions that can list decode from $1-\eps$ fraction of errors with an optimal list size of $O(1/\eps)$. Fully explicitly, we construct a code with rate $\Omega(\varepsilon^2)$ and alphabet size $2^{\poly \log(1/\eps)}$. While the rate here is sub-optimal, it is still the best known in any explicit construction with list size $O(1/\eps)$, and it significantly improves the alphabet size in the previously best known such explicit code, which is $2^{O(1/\eps \log(1/\eps))}$ \cite{guruswami2001expander}\footnote{We note that \cite{guruswami2001expander} also has a non-explicit construction giving the same parameters as ours. Furthermore, such parameters are easy to achieve by the Johnson bound from any code with relative distance $1-\eps^2$, but this does not give polynomial time decoding algorithms.}. Next, we present a method to achieve near-optimal list size and rate simultaneously, by introducing a generalization of standard disperser which we call a \emph{multi-set disperser}. We show that given an optimal multi-set disperser as an ingredient in our graph based construction, we can obtain a $(1 - \varepsilon, O(\frac{\log^2(1/\varepsilon)}{\varepsilon}))$-list decodable code with rate $\Omega(\frac{\varepsilon}{\log^2(1/\varepsilon)})$ and alphabet size $\poly(1/\eps)$. Hence, all the parameters are close to optimal. We note that this construction is semi-explicit, as in many previous works \cite{guruswami2001expander, guruswami2002near, guruswami2004better, rom2006improving}, since such a multi-set disperser can be constructed in probabilistic polynomial time (a random graph is such a disperser with high probability). However, the semi-explicit construction in \cite{guruswami2001expander} which achieves similar rate and list size requires an alphabet size of $2^{O(1/\varepsilon\log(1/\varepsilon))}$. In addition, the semi-explicit construction in \cite{guruswami2001expander} does not give a polynomial-time list decoding algorithm (even assuming the optimal non-explicit object). In contrast, we provide a polynomial time list decoding algorithm for all our codes. Specifically, we have the following two theorems.

\begin{theorem}\label{thm: intro4}
    Given any $\varepsilon > 0$, there exists an explicit family of codes over an alphabet of size $\quasipoly\left(\frac{1}{\varepsilon}\right)$, which have rate $\Omega(\varepsilon^2)$ such that a code with block length $n$ in this family can be list decoded up to $(1-\varepsilon)$ fraction of errors with list size $L = O(\frac{1}{\varepsilon})$ in time $\poly_{\varepsilon}(n)$.
\end{theorem}

\begin{definition}
A $(K,\delta)$-multi-set disperser is a bipartite graph $G = (L\sqcup R, E)$, where $|L(G)|=N$ and $|R(G)|=M$, with a uniform left degree denoted as $D$, which has the property that for every $t \geq 2$ and $t$ different subsets $\mathcal{S}_1, \ldots, \mathcal{S}_t$ of $R(G)$ with $|\mathcal{S}_i \Delta \mathcal{S}_j| \geq \delta\cdot |R(G)|$ for each $i$ and $j$, the number of $v \in L(G)$ such that $\Gamma(v)$ is contained in some $(\mathcal{S}_i \Delta \mathcal{S}_j)^c$ is at most $t \cdot K$.
\end{definition}

There is a mutual conversion relationship between dispersers and multi-set dispersers. Specifically, a $(K, \delta)$-multi-set disperser is a $(2K, \delta)$-disperser. Conversely, a $(K, \delta)$-disperser is a $(\sqrt{\frac{KN}{2}}, \delta)$-multi-set disperser. However, a disperser with constant entropy loss does not necessarily imply a good multi-set disperser, where we aim to achieve $M/KD = \Omega(1)$. The existence of such a multi-set disperser can be proven by standard probabilistic argument, and we have the following theorem.

\begin{theorem}\label{thm: intro5}
For every integers $N>0, 0< K < N/2$, and $0<\delta<1$, there exists an $(K, \delta)$-multi-set dispersers $G: (L\sqcup R, E)$, where $|L|=N$, $|R|=M$, with left degree $D$ and
\begin{itemize}
    \item $M = \lfloor K/N\log(N/K)\rfloor\cdot N$
    \item $D = \lceil 4\log(N/K)/\delta \rceil$.  
\end{itemize}
Moreover, a random bipartite graph with $N$ left vertices, $M$ right vertices, and left-degree $D$ as specified above above is a $(K, \delta)$-multi-set disperser with probability at least $1 - \exp(1 - \log(N/K)N)$.
\end{theorem}

Based on that, we provide a semi-explicit list decodable code with near-optimal list size and near-optimal rate.

\begin{theorem}\label{thm: intro6}
For any $\varepsilon > 0$, there exists a family of codes over an alphabet of size $\poly(1/\varepsilon)$, which has rate $\Omega\left(\frac{\varepsilon}{\log^2(1/\varepsilon)}\right)$ such that a code with block length $n$ in the family can be list decoded from up to $1-\varepsilon$ fraction of errors with list size $L = O\left(\frac{\log^2(1/\varepsilon)}{\varepsilon}\right)$ in time $\poly_{\varepsilon}(n) $, and can be constructed in probabalistic time $\poly(n,\log(1/\varepsilon))$ with success probability at least $1-\exp(-\log(1/\varepsilon)\cdot n)$.
\end{theorem}

We point out that this is the first near-optimal, efficiently list-decodable code in the high-noise regime that achieves the optimal list size with a polynomial-sized alphabet. It is worth noting that this result does not contradict the lower bound provided in \cite{alrabiah2024ag}, where $q \geq \exp(1/\eps)$, because the rate of our code is a function of $\eps$. We refer the reader to Table \ref{tab_list} for a detailed comparison.
\begin{table}[h!]
    \centering
    \renewcommand{\arraystretch}{1.5}
    \begin{tabular}{
    | >{\centering\arraybackslash}m{2cm} 
    | >{\centering\arraybackslash}m{1.8cm} 
    | >{\centering\arraybackslash}m{4cm} 
    | >{\centering\arraybackslash}m{3cm} 
    | >{\centering\arraybackslash}m{2.2cm}
    | >{\centering\arraybackslash}m{2cm}
    |
    }
    \hline
   & \textbf{Rate} $R$  & \textbf{Alphabet size} &  \textbf{List size} & \textbf{Poly-time construction} & \textbf{Decoding time} \\
   \hline
    \cite{guruswami2001expander}  & $\varepsilon^2$ & $\exp(1/\varepsilon\log(1/\varepsilon))$ & $\frac{1}{\varepsilon}$&  D &  $n^2$ \\
    \hline
    \cite{guruswami2001expander}  & $\varepsilon$ & $\exp(1/\varepsilon\log(1/\varepsilon))$ & $\frac{1}{\varepsilon}$ &  P &  $\exp(n^{\gamma})$ \\
    \hline
    \cite{guruswami2002near}  & $\varepsilon^2$ & $\exp(\poly\log(\frac{1}{\varepsilon}))$ & $\frac{1}{\varepsilon}$ & P &  $n^2\poly\log n$ \\
    \hline
    \cite{guruswami2002near}  & $t^{-3}\varepsilon^{2+\frac{2}{t}}$ & $1/\varepsilon^b$ ($b>t$) & $\frac{t^2}{\varepsilon^{1+1/t}}$ & P &  $n^{1/\varepsilon}$ \\
    \hline
    \cite{guruswami2004better}  & $\frac{\varepsilon}{\poly\log(1/\varepsilon)}$ & $\exp(1/\varepsilon\log(1/\varepsilon))$ & $\exp(\sqrt{n\log n})$ & D &  $\exp(n^{1/2})$ \\
    \hline
    \cite{rom2006improving} & $\frac{\varepsilon}{\poly\log(1/\varepsilon)}$ & $\exp(\exp(\poly\log\log(\frac{1}{\varepsilon})))$ & $\exp(\sqrt{n\log n})$ &  D & $\exp(n^{1/2})$ \\
    \hline
    \cite{rom2006improving} & $\frac{\varepsilon}{\poly\log(1/\varepsilon)}$ & $\exp(\poly\log(1/\varepsilon))$ & $\exp(\sqrt{n\log n})$ &  P & $\exp(n^{1/2})$ \\
    \hline
    \cite{guo2021efficient}&$\varepsilon$&$\exp(1/\varepsilon^2\log(1/\varepsilon))$&$\exp(\poly(1/\varepsilon))$& D &$\poly(n)$ \\
    \hline
    \cite{ron2024efficient}&$\varepsilon$&$\poly_{\epsilon}(n)$&$\frac{1}{\varepsilon}$& P &$\poly(n)$ \\
    \hline

    \cite{chen2024explicit}&$\varepsilon$&$\poly_{\epsilon}(n)$&$\frac{1}{\varepsilon}$& D &$\poly(n)$ \\
    \hline
    
    Our Work  & $\varepsilon$ & $\exp(\poly\log(\frac{1}{\varepsilon}))$ & $\exp\exp\exp(\log^{\ast}n)$ &  D &  $\poly(n)$ \\
    \hline
    Our Work  & $\varepsilon^2$ & $\exp(\poly\log(\frac{1}{\varepsilon}))$ & $\frac{1}{\varepsilon}$ &  D &  $\poly(n)$ \\
    \hline
    Our Work  & $\frac{\varepsilon}{\log^2(1/\varepsilon)}$ & $\poly(1/\varepsilon)$ & $\frac{\log^2(1/\varepsilon)}{\varepsilon}$ &  P &  $\poly(n)$ \\
    \hline
    \end{tabular}
    \caption{List decodable codes capable of decoding a $(1-\varepsilon)$ fraction of errors, where $\gamma$ is a constant affecting the rate by a constant factor. The fifth column indicates whether such codes can be obtained in polynomial time, either deterministically or probabilistically. We use D for deterministic polynomial time and P for probabilistic polynomial time. In order to keep things concise here, we have omitted the $O(\cdot)$ and $\Omega(\cdot)$ notations. When presenting the decoding algorithm, we treat $\varepsilon$ as a constant.}
    \label{tab_list}
\end{table}

\subsection{Proof Overview}

At a high level, our constructions are based on using certain extremal graphs to achieve the desired properties from a weaker code, that is easier to construct. Such constructions are widely used in previous works (e.g., \cite{5feff923592b4956b155bd935ccfd8c5,guruswami2001expander,guruswami2002near,guruswami2004better,rom2006improving}). Here, we use modified constructions and provide new decoding algorithms. Our constructions require two fundamental components: a bipartite graph (usually a variant of \emph{disperser}) and a weaker code to start with (often called the \emph{mother code}).

\subsubsection*{Graph-Concatenated Code}

Dispersers and expanders are often used to amplify the distance or achieve other desired properties of a code, together with a well designed mother code. In the following, we represent a disperser by a bipartite graph $G: (L\sqcup R,E)$.

Consider a bipartite graph $G$ with $N$ left vertices and $M$ right vertices, each left vertex having a uniform degree $D$. Starting with a mother code $\mathcal{C}$ of length $M$ over an alphabet $\Sigma$, we construct a graph-concatenated code $G(\mathcal{C})$ of length $N$ over the alphabet $\Sigma^D$. This construction is carried out by placing each codeword $c\in\mathcal{C}$ on the right side of the graph, that is, each right vertex corresponds to a coordinate in $c$. Each left vertex then receives $D$ symbols from all its neighbors, which are juxtaposed to form a new symbol. Consequently, each position in the concatenated code $G(\mathcal{C})$ corresponds to a left vertex, forming a new codeword $G(c)$. The complete set of these new codewords form the code $G(\mathcal{C})$, which maintains a one-to-one correspondence with the original codewords in $\mathcal{C}$. Importantly, the linearity of $\mathcal{C}$ is inherited in $G(\mathcal{C})$; that is, if $\mathcal{C}$ is $\mathbb{F}$-linear over some field $\mathbb{F}$, then $G(\mathcal{C})$ remains $\mathbb{F}$-linearity over a larger field.

\subsubsection*{Distance Amplification}
Rather than using expansion properties of the graph as in previous works, we directly argue about the distance of $G(\mathcal{C})$ by using the disperser property. To demonstrate how the distance of the code $G(\mathcal{C})$ is amplified by the disperser $G$, we analyze the distance between each pair of codewords in $G(\mathcal{C})$ by examining the corresponding codewords in $\mathcal{C}$.

Recall that a $(K, \delta)$-disperser has the property that for every subset of left vertices $S \subseteq L(G)$ of size at most $K$, the size of its neighbors $|\Gamma(S)|$ is at least $(1 - \delta)|R(G)|$. Consider a code $\mathcal{C}$ on the right with distance $\delta M$. For any two codewords $c^1, c^2$ in $\mathcal{C}$, at most a $(1-\delta)$ fraction of the vertices are identical.  Using the disperser property, we can bound the number of left vertices whose neighbors fall entirely into this set. Specifically, if $G$ is a $(K,\delta/2)$-disperser, the subset of left vertices where $G(c^1)$ and $G(c^2)$ have identical values is less than $K$. This implies that the corresponding codewords $G(c^1)$ and $G(c^2)$ in $G(\mathcal{C})$ have less than $K$ identical symbols, and hence the distance of $\mathcal{C}$ is at least $N - K$.

The rate of the final code is $\frac{R(\mathcal{C})M}{ND} = \frac{R(\mathcal{C})}{\delta} \cdot \frac{K}{N} \cdot \frac{\delta M}{KD}$, and our goal is to achieve rate $\Omega(\varepsilon)$. Therefore, we set $K/N = \Theta(\varepsilon)$, and both $\frac{R(\mathcal{C})}{\delta}$ and $\frac{\delta M}{KD}$ to be constants independent of $K/N$. To ensure $\frac{R(\mathcal{C})}{\delta}$ is a constant, we simply use an asymptotically good code as the mother code, since this implies that both $R(\mathcal{C})$ and $\delta$ are absolute constants. To ensure $\frac{\delta M}{KD}$ is a constant, we use a disperser with constant entropy loss. This is possible since we can set the error of the disperser to be $\delta/2$ as discussed before, which is an absolute constant. If we use the explicit disperser in \cite{capalbo2002randomness}, this gives Theorem~\ref{thm:intro1}, while using an optimal (but non-explicit) disperser gives Theorem~\ref{thm:intro2}.

Compared to previous works, our construction uses a better unbalanced disperser as given in \cite{capalbo2002randomness}. This improves upon the balanced disperser/expander in \cite{guruswami2001expander} and the unbalanced disperser with poor degree in \cite{feldman2006lp}. This gives us the improvement over previous results.

\subsubsection*{Unique Decoding Algorithm}

For similar constructions of graph-based codes, \cite{guruswami2001expander} presents two algorithms for unique decoding, where the decoding can be interpreted as a voting operation on the right side, with each vertex being assigned the majority symbol from its neighbors. The first algorithm, which applies to codes achieving a rate of $\Omega(\varepsilon)$, requires the mother code to be list-decodable. This is because, after the voting operation, we cannot guarantee that the correct bits in the mother code constitute a majority, so list decoding is needed to generate a filtered list that can be checked individually. As a result, the decoding time is superlinear. The second algorithm achieves linear-time decoding since, after the voting process, it runs the linear-time decoding algorithm of the mother code, based on the expander code from \cite{sipser1996expander}. However, this approach requires the graph to possess certain additional properties, specifically that it is derived from a Ramanujan graph with $D = O(\frac{1}{\varepsilon^2})$. The Ramanujan graph ensures that most right-side vertices vote for the correct bits. Nonetheless, these strict graph requirements result in a worse rate ($\Omega(\varepsilon^2)$) and larger alphabet size ($2^{O(1/\eps^2)}$). The construction and unique decoding algorithm in \cite{rom2006improving} build on work in \cite{guruswami2001expander}, by transitioning from a balanced bipartite graph to an unbalanced one. While this improves the alphabet size, the unique decoding algorithm still depends on the list decoding algorithm of the mother code, and hence takes super linear time. The linear-time unique decoding algorithm in \cite{guruswami2002near}, on the other hand, use a new construction by considering each right vertex as a Reed-Solomon code with length equal to the right degree of the graph. However, the disadvantage of this construction is that it requires a large alphabet size for the mother code and a large degree for the graph, resulting in an exponentially large alphabet size in the final code.

Here we develop a new unique decoding algorithm for our near-optimal codes, by introducing erasure operations for the received word. Specifically, upon receiving a word $y$, we distribute all values from the left side to the right, in the graph $G$. Our next crucial step involves identifying any inconsistent values. Specifically, if there are two vertices on the left, $u$ and $v$, and indices $i, j \in [D]$, where $\Gamma_i(u) = \Gamma_j(v)$ for the neighbors on the right (i.e., the $i$'th neighbor of $u$ and the $j$'th neighbor of $v$ are the same vertex on the right) but $y(u)_i \neq y(v)_j$ for the corresponding values, this indicates inconsistent values on the right side and therefore there must be at least one error. In this case, we replace the values on $u$ and $v$ by a symbol $\perp$, i.e., treat them as erasures. This process continues until, for each vertex $w$ on the right, either every vertex in $\Gamma(w)$ is $\perp$ or all non-$\perp$ values in $\Gamma(w)$ are consistent. Finally, if there are real erasures in the received word $y$ which results in `empty' right vertices, we just fill them with arbitrary values. We can now apply the linear time unique decoding algorithm of the mother code $\mathcal{C}$, which is again based on expander codes.

We show that our decoding algorithm succeeds as long as the condition $2e + s < (1 - \varepsilon)n$ is satisfied, where $e$ is the number of errors that occurred and $s$ is the number of erasures that occurred. Note that in each iteration where we replace a pair of vertices on the left by $\perp$, at least one of these vertices is erroneous. Here, we set the code length $n$ to be equal to the number of left vertices $N$. Since $e<(1 - \varepsilon)n/2$, after the iterations, at least $n-2e=\eps n=K$ left vertices remain. Let $\mathcal{S}$ denote the set of these vertices. By the disperser property, we know that $|\Gamma(\mathcal{S})| \geq (1 - \delta/2)M$, and each vertex $v \in \Gamma(\mathcal{S})$ receives at least one correct value from the left. Consequently, the proportion of correct symbols on the right is at least $1 - \delta/2$. Therefore, the unique decoding algorithm for the mother code can be effectively applied.

We note that as long as the mother code is uniquely decodable in linear time, the whole process can be implemented in linear time, since the graph has constant degree. Therefore, we employ expander codes as mother codes, which possess unique decoding algorithms up to a sufficient decoding radius. 

\subsubsection*{List Decodable Codes and Decoding Algorithm} Our list decodable codes are also based on using a graph together with a mother code. In particular, we use a list recoverable code as the mother code. Informally, a list recoverable code has the property that if we know each symbol of the codeword is contained in a small list of size $\ell$, then the total number of such possible codewords is also small. List recoverable codes have been widely used as ingredients to construct list decodable codes. Here, our construction is similar in spirit to that of \cite{guruswami2004better} and the follow up work \cite{rom2006improving}, where the mother code is a so called extractor code, which has good list recovery properties. In fact, list recoverary from extractor codes does not need each candidate symbol to be contained in a small list. Instead, it suffices if the total number of all candidate symbols and their corresponding positions is bounded. However, the drawback is that the currently best known explicit constructions result in a large, sub-exponential output list size. Here, we use a more general list recoverable code as the mother code and provide an enhanced algorithm, allowing for efficient list decoding with better parameters.

Specifically, our first construction uses a list recoverable code in \cite{hemenway2019local} as the mother code. We set the list size, $\ell$, to $\frac{2}{\varepsilon}$ and maintain both the rate and list recovery radius as constants. The graph, which is a $(K, \delta)$-disperser, satisfies $K/N = \frac{\varepsilon}{2}$, and $\delta$ is no more than the list recovery radius of the mother code.

Upon receiving a word $y$, the initial step involves distributing all values from the left side to the right side. On average, each right vertex receives $\frac{ND}{M}=O(1/\varepsilon)$ values. Given that this graph is not necessarily bi-regular, some right vertices may receive a lot more values. Thus, we adopt our strategy from unique decoding, by deleting some inconsistent values. To manage the list size for list recovery effectively, we ensure that no right vertex receives more than $\ell$ different values; otherwise, we eliminate $\ell$ different values and erase the corresponding positions on the left, since we know at most one of the values is correct. We set $\ell$ to be $\frac{2}{\varepsilon}$ so that the total number of correct vertices removed during this process is at most $\frac{\varepsilon}{2}$ fraction. Consequently, at least $\frac{\varepsilon}{2} N$ correct left vertices remain. By the disperser property, these vertices span at least $(1-\delta)$ fraction of the right vertices, sufficiently covering the required positions for the list recovery algorithm.

We note there may be cases where a left vertex has parallel edges to the same right vertex but distributes different values. If this happens, we know that this left vertex is erroneous and erase it in the first step.

After obtaining a list from the list recovery algorithm of the mother code, we also get a list of the codewords of the new code, just by running the encoding algorithm. We then individually check the distance between the candidate codewords and the received word to obtain an appropriate list.

\subsubsection*{List Decoding with Near Optimal List Size}

However, with the above approach it will be difficult for the new code to achieve an optimal list size, as this would require a very strong list recoverable mother code. To address this, we consider a stronger property known as `average-radius list decoding', as defined in \cite{rudra2014every}. Specifically, we say a code $\mathcal{C}$ is $(1-\varepsilon, L-1)$-average-radius list decodable, if the following holds: For any set of $L$ different codewords, denoted $\Lambda$, we have
\begin{equation}
\sum_{j \in [n]} \mathbf{pl}_j(\Lambda) < \varepsilon nL,
\end{equation}
where $\mathbf{pl}_j(\Lambda)$ denotes the plurality at the $j$'th coordiante (i.e., the number of the symbol which appears the most). It is important to notice that a $(\rho, L)$-average-radius list decodable code is also $(\rho, L)$-combinatorially list decodable. Therefore, to prove that a code is list decodable, we only need to bound $\sum_{j \in [n]} \mathbf{pl}_j(\Lambda)$ for any $L$ different codewords.

We first start with a preliminary result, using plurality to prove a simple conclusion and laying the groundwork for subsequent proofs. Fix an asymptotically good mother code $\mathcal{C}$ with constant relative distance and constant rate. We show that when the rate of the concatenated code $G(\mathcal{C})$ provided by the disperser is $\Omega(\varepsilon^2)$, the resulting code can list decode from a $1 - \varepsilon$ fraction of errors with an optimal list size of $O(1/\eps)$. To illustrate this, we examine any two different codewords $c^1, c^2$ in $G(\mathcal{C})$. We prove that the number of left vertices where $G(c^1)$ and $G(c^2)$ share the same value does not exceed $ O(\varepsilon^2 n) $. Finally, we use a double counting method to show that the sum of the plurality does not exceed $O(\varepsilon^2 L^2 n) = O(\varepsilon L n)$ by taking $L=O(1/\varepsilon)$.

The above argument merely matches the Johnson bound, and is too loose to achieve a near-optimal rate. To address this, we analyze a more intricate property which we call the plurality condition. Specifically, we say that a code $\mathcal{C}$ satisfies the $(\beta, \delta, L)$-plurality condition if, for any $L$ different codewords $\Lambda$ and a uniformly randomly chosen index $j$, the probability that $\mathbf{pl}_j(\Lambda) \geq \beta L$ is at most $\delta$. Setting $\beta = \delta = \Theta(\varepsilon)$ leads to a feasible $(1 - \varepsilon, L)$-list decodable code, though the rate of a random code satisfying this property is only $\Omega(\varepsilon^2)$. To overcome this barrier, we show the following modified property: if a code satisfies the $(\beta, \delta, L)$-plurality condition for all $\beta, \delta$ such that $\beta \cdot \delta = \varepsilon$, then this code is $(1 - O(\varepsilon \log(L)), L)$-list decodable. by choosing $L = O(1/\varepsilon)$, the list decoding radius is within a logarithmic factor of the optimal. Moreover, our crucial observation is that, implicitly, a random code over a large enough alphabet with rate $\Omega(\varepsilon)$ satisfies the $(\beta, \delta, L)$-plurality condition for all $\beta$ and $\delta$ such that $\beta \cdot \delta = \varepsilon$ with $L = O(1/\varepsilon)$ with high probability.

However, to satisfy this property using graph concatenated codes is non-trivial. Specifically, it is hard to show that dispersers with $K=\Omega(\varepsilon N)$ are enough give such codes where $\beta\cdot \delta =\varepsilon$. Therefore, we introduce a generalization and strengthening of standard dispersers, which we call \textit{multi-set dispersers}, and show that such objects can be used to achieve both near-optimal rate and near-optimal list size. The key part of our proof is establishing a connection between the multi-set disperser and the plurality condition. Furthermore, we show that a random graph satisfies the desired multi-set disperser property with high probability. 

We replace the disperser in the construction with a multi-set disperser $G$, and start from an asymptotically good binary mother code $\mathcal{C}$ with constant relative distance $\delta$. For each $L$ different codewords $\Lambda\subseteq\mathcal{C}$, define $\mathcal{S}_i$ as the indices where the $i$-th codeword equals 1. To bound the size of the set $\mathcal{Q} \subseteq L(G)$ where the vertices in $\mathcal{Q}$ have plurality at least $\beta L$ for $\Lambda$, we claim that we can select $\lceil \frac{2}{\beta} \rceil $ codewords $\mathcal{T}$ from $\Lambda$ and a subset $\mathcal{Q}^{\prime} \subseteq \mathcal{\mathcal{Q}}$ with size at least $0.3|\mathcal{Q}|$, such that for each $v \in \mathcal{Q}^{\prime}$, we can always find two codewords from $\mathcal{T}$ whose values are identical at $v$. This implies that all neighbors of this vertex are contained in $(\mathcal{S}_i \Delta S_j)^c$ for some $c^i, c^j \subseteq \mathcal{T}$. By the multi-set disperser property, we can bound the size of $\mathcal{Q}^{\prime}$, and thus also the size of $\mathcal{Q}$. This gives our near optimal list decodable code which can be constructed in probabilistic polynomial time.

In this construction, the plurality condition analysis does not give us a decoding algorithm. To get such an algorithm, we need to use again the list decoding algorithm from our previous list decodable code. However, we cannot simply replace the mother code with the previous one since, in our new construction, the mother code for the multi-set disperser needs to be binary. To address this, we introduce a folding operation that combines two different codes with the same message space and codeword length. This operation allows the resulting code to inherit all decoding properties of the two original codes. For the first code, we use an efficient list decodable code, but the list size is not optimal. For the second code, we use a code that is list decodable with a near-optimal list size, but no decoding algorithm. Putting them together, the folded code has both a near-optimal list size and an efficient decoding algorithm.

\subsection*{Outline of the Paper}
The rest of the paper is organized as follows. In Section \ref{sec_pre}, we provide the necessary notations and definitions. In Section \ref{sec_cons}, we present the framework of our construction. In Section \ref{sec_uni}, we prove Theorem \ref{thm:intro1} and Theorem \ref{thm:intro2} for linear-time uniquely decodable codes. In Section \ref{sec_list}, we prove Theorem \ref{thm: intro3} for efficiently list decodable codes. In Section \ref{sec_ran}, we prove Theorem \ref{thm: intro4}, introduce the definition of multi-set dispersers, and prove Theorems \ref{thm: intro5} and \ref{thm: intro6}. In Section \ref{sec_open}, we pose some open problems.

\section{Preliminary}\label{sec_pre}

We employ standard Landau notation $ O(\cdot), \Omega(\cdot), \Theta(\cdot) $ to denote the asymptotic order of a function. Additionally, we use $ \poly(\cdot) $ to refer to polynomial functions, $ \exp(\cdot) $ to refer exponential growth functions, and $ \quasipoly(\cdot) $ is defined as $ \exp(\polylog(\cdot)) $. We introduce  $O_x(\cdot)$, $\Omega_x(\cdot),\Theta_x(\cdot),\poly_x(\cdot)$ and $\exp_x(\cdot)$ to specify consideration with respect to a fixed factor $ x $. All logarithms are base-2 unless stated otherwise. The iterated logarithm of a number $n$, denoted as $\log^{\ast}(n)$, is defined as the number of times the logarithm function must be applied to $n$ before the result is less than or equal to $1$. For integer $ n $, let $ [n] $ denote the set $ \{1, \ldots, n\} $, and $ \Sigma^n $ denotes the set of all strings of length $n$ over alphabet $\Sigma$.
For a string $ s \in \Sigma^n $ and an index $ i \in [n] $, $ s_i $ denotes the $ i $-th symbol of $ s $.

When the context is clear, we use the notation $ \mathcal{S}^c $ to denote the complement of the set $ \mathcal{S} $, and $ \mathcal{S}_i \Delta \mathcal{S}_j $ to denote the symmetric difference of two sets $ \mathcal{S}_i $ and $ \mathcal{S}_j $. Specifically, $ \mathcal{S}_i \Delta \mathcal{S}_j = (\mathcal{S}_i \backslash \mathcal{S}_j) \cup (\mathcal{S}_j \backslash \mathcal{S}_i) $.

Given a binomial coefficient $ \binom{N}{K} $, we will utilize the inequality
\begin{equation}
\binom{N}{K} \leq \left( \frac{e N}{K} \right)^K \end{equation} 
in the following proofs. Additionally, we employ the Chernoff bound through the following lemma:

\begin{lemma}
    Let $ X = \sum_{i=1}^n X_i $, where $ X_i = 1 $ with probability $ p_i $ and $ X_i = 0 $ with probability $ 1 - p_i $, and all $ X_i $ are independent. Let $ \mu = \mathbb{E}(X) = \sum_{i=1}^n p_i $. Then
    \begin{equation}
    \Pr[X > (1+\delta)\mu] < \left( \frac{e^\delta}{(1+\delta)^{1+\delta}} \right)^\mu < \left( \frac{e}{1+\delta} \right)^{(1+\delta)\mu}
    \end{equation} 
    for any $ \delta > 0 $.
\end{lemma}

A code $\mathcal{C} \subseteq \Sigma^n$ consists of words of equal length $n$ over a fixed alphabet $\Sigma$. The dimension of a code $\mathcal{C}$, denoted $k(\mathcal{C})$, is defined as $k(\mathcal{C}) := \log_{|\Sigma|} |\mathcal{C}|$, and its rate $R(\mathcal{C})$ is given by $R(\mathcal{C}) := \frac{k(\mathcal{C})}{n} = \frac{\log_{|\Sigma|} |\mathcal{C}|}{n}$. The distance of a code $\mathcal{C}$, denoted $d(\mathcal{C})$, is defined as $d(\mathcal{C}) := \min_{c^1 \neq c^2 \in \mathcal{C}} d(c^1, c^2)$, where $d(\cdot, \cdot)$ represents the Hamming distance between two words, and the relative distance $\delta(\mathcal{C})$ is given by $\delta(\mathcal{C}) = \frac{d(\mathcal{C})}{n}$.

A code with distance $d$ can decode up to $\left\lfloor \frac{d}{2} \right\rfloor - 1$ errors. That is, for any $y \in \Sigma^n$, there is at most one codeword $c \in \mathcal{C}$ such that $d(y, c) < \frac{d}{2}$. In coding theory, there is another type of corruption called erasures, where some symbols in $y$ are replaced by a symbol not in $\Sigma$, denoted as $\perp$. Considering both errors and erasures, a code with distance $d$ can decode up to $e$ errors and $s$ erasures whenever $2e + s < d$. We often refer to $2e + s$ as the number of `half errors', where one erasure counts as one half error and one error counts as two half errors.

A code $\mathcal{C} \subseteq \Sigma^n$ is called $(\rho, L)$-(combinatorially) list-decodable if, for every $y \in \Sigma^n$, the number of codewords $c \in \mathcal{C}$ within a Hamming distance $\rho n$ from $y$ is at most $L$. Furthermore, a code $\mathcal{C} \subseteq \Sigma^n$ is $(\rho, L)$-average-radius list decodable if, for any $L+1$ distinct codewords $c^{1},c^{2}, \ldots, c^{L+1} \in \mathcal{C}$ and any vector $y \in \Sigma^n$, the average Hamming distance from $y$ to $c^{1},c^{2} \ldots, c^{L+1}$ is strictly greater than $\rho n$. It is noteworthy that average-radius list decoding is stronger than combinatorial list decoding: any $(\rho, L)$-average-radius list decodable code is also $(\rho, L)$-combinatorially list decodable. A code $\mathcal{C} \subseteq \Sigma^n$ is said to be $(\rho, \ell, L)$-list recoverable if, for every collection of $\ell$ sets $S_1, S_2, \ldots, S_n \subseteq \Sigma$ and every received vector $y \in \Sigma^n$, the number of codewords $c \in \mathcal{C}$ such that $c_i \in S_i$ for at least $(1 - \rho)n$ positions is at most $L$. List-recoverable codes generalize list-decodable codes by allowing the received symbols to belong to sets of possible values rather than a single value.

A disperser is a combinatorial structure designed to ensure that for any subset of a given size, the output size remains high when passed through the disperser. Essentially, dispersers are utilized to extract randomness from weakly random sources, thus reducing the predictability of the output. Typically, a disperser can be represented using a bipartite graph as follows.

\begin{definition}
A $(K, \delta)$-disperser is a bipartite graph $G : (L \sqcup R, E)$ with a uniform left degree denoted as $D$, which has the property that for every subset $S \subseteq L(G)$ with size at least $K$, its neighbors $\Gamma(S)$ have size at least $(1 - \delta)|R|$. $\delta$ is called the error of this disperser, and $\log\left(\frac{KD}{M}\right) = \log(K) + \log(D) - \log(M)$ is the entropy loss of this disperser.
\end{definition}

The pursuit of constructing better dispersers has long been a goal in much research. Generally, given $ N, K, \delta $, we have two main objectives: one is to reduce the degree $D$, and the other is to minimize the entropy loss. Optimally but not explicitly constructed, we have the following parameter scheme for a disperser.

\begin{lemma}\label{lemma_disp}
    For every integers $N>0, 0< K < N$, and $0<\delta<1$, there exists a $ (K, \delta) $-disperser $G: (L\sqcup R, E)$ with left degree $D$, where $|L|=N$, $|R|=M$ with
    \begin{itemize}
        \item $D = O(\log(N/K)/\delta)$, and 
        \item $M = \Theta(KD/\log(1/\delta))$.
    \end{itemize}
    Moreover, a random bipartite graph with $N$ left vertices, $M$ right vertices, and left-degree $D$ as specified above above is a $(K, \delta)$-disperser with probability at least $1 - \exp(1 - \log(N/K)N)$.
\end{lemma}

The existence of dispersers with these parameters can be proven by probabilistic arguments. However, current explicit constructions still have significant room for improvement in both aspects. With a fixed constant $ \delta $, several constructions have achieved constant entropy loss, while the degree is a polynomial function of $\log(N/K)$. In \cite{capalbo2002randomness}, the authors unify various randomness objects into a framework known as randomness conductors and introduce the concept of `extracting conductors', which is a stronger version of extractors. Note that a random extractor has stronger properties than a disperser. In this paper, we just employ the `disperser' property of such objects, and the results can be presented as follows.

\begin{theorem}\cite{capalbo2002randomness}\label{theorem_disp}
For every integers $N>0, 0< K < N$, and $0<\delta<1$, there exists a $ (K, \delta) $-disperser $G: (L\sqcup R, E)$ with left degree $D$, where $|L|=N$, $|R|=M$ with
    \begin{itemize}
        \item $D = \poly(\log(N/K),1/\delta)$, and 
        \item $M = \Theta(KD/\delta^2)$.
    \end{itemize}
Moreover, if $K=\Theta(\alpha N)$ for some constant $\alpha$, then this disperser can be constructed in time $\poly_{\alpha}(N)$.
\end{theorem}

\section{Graph-Concatenated Code}\label{sec_cons}
In this section, we introduce how to construct a new code $ G(\mathcal{C}) $ from a bipartite graph $ G: (L \sqcup R, E) $ with left degree $ D $, and a code $ \mathcal{C} \subseteq \Sigma^{R} $, which is defined on the right vertices of $G$. This construction was first introduced in \cite{guruswami2001expander}, and subsequently developed in \cite{guruswami2002near, guruswami2004better, rom2006improving}. We adopt their framework here.

For each codeword in $\mathcal{C}$, which can be represented by a function $c: R \mapsto \Sigma$, we  concatenate $c$ with graph $G$, denoted as $G(c): L \mapsto \Sigma^D$, which is given by
\begin{equation}
G(c)(l) = \left( c\left(\Gamma_1(l)\right), c\left(\Gamma_2(l)\right), \ldots, c\left(\Gamma_D(l)\right)  \right),
\end{equation}
for each $l\in L$, where $\Gamma_i(l) \in R$ represents the $i$-th neighbor of $l\in L$. The concatenated code $G(\mathcal{C})$ is then defined as:
\begin{equation}
G(\mathcal{C}) = \{ G(c) | c \in \mathcal{C} \}.
\end{equation}

We observe that this construction distributes the values of each codeword through the bipartite graph from the right side to the left side, forming new codewords on the left. Because each value is transmitted multiple times to the left side, this results in a reduction in rate and an increase in the alphabet size. Specifically,  the resulted code $G(\mathcal{C})$ is over alphabet $\Sigma^D$ and of length $|L|$. The rate of $G(\mathcal{C})$ can be computed as $\frac{R(\mathcal{C})|R(G)|}{D|L(G)|}$. Moreover, if $\mathcal{C}$ is a $\mathbb{F}_q$ -linear code for some prime power $q$, then $G(\mathcal{C})$ is also $\mathbb{F}_q$ -linear. At the same time, this enhances the robustness of the code, leading to strong capabilities in unique decoding and list decoding, which we will discuss in the following sections.

\section{Linear-Time Uniquely Decodable Code}\label{sec_uni}

In this section, we present a linear-time uniquely decodable code with optimal rate-distance trade-off based on graph-code concatenation. We begin with a mother code $\mathcal{C}$, which is a binary code with a constant rate and desirable distance. To achieve a linear-time decoding algorithm, we start with a code that is uniquely decodable in linear time. Therefore, we employ expander codes, which have shown significant progress in recent years. Specifically, we utilize the explicit expanders presented in \cite{golowich2024new} and the expander code decoding analysis provided in \cite{viderman2013linear}.

\begin{definition}
For real number $0 < \varepsilon < 1$, a bipartite graph $G: (L \sqcup R, E)$ with $N$ left vertices, $M$ right vertices, and uniform left-degree $D$ is called a $(K, (1-\varepsilon)D)$-lossless expander for some $K< N$, if for every subset $S \subseteq L$ with size at most $K$, it holds that $\left|\Gamma(S)\right| \geq (1-\varepsilon) D |S|$.
\end{definition}

The graph $G$ is called $(D, C)$-biregular if every left vertex has degree $D$ and every right vertex has degree $C$, with the constraint $ND=MC$. If we can achieve arbitrarily small $ \varepsilon $, we term this a `lossless expander'. In particular, explicit constructions of lossless expanders with small degrees $ D $ and small expansion cutoff $ \alpha=K/N$ are always of interest. Recently, \cite{golowich2024new} provided an explicit construction of constant-degree lossless expanders with an arbitrary ratio between the sizes of the left vertex set and the right vertex set, presented as follows:

\begin{lemma}[\cite{golowich2024new}]
For every $\beta, \varepsilon > 0$, there is an infinite explicit family of $(\alpha N, (1-\varepsilon) D)$-lossless expanders $G$ with $|R(G)| / |L(G)| > \beta$, and with $$D = O\left(\frac{(\log(1/\varepsilon) + \log(1/\beta))^2}{\varepsilon^3}\right)$$ and $$\alpha = \Omega\left(\frac{\varepsilon^4 \beta}{\left(\log \frac{1}{\varepsilon} + \log \frac{1}{\beta}\right)^2}\right).$$
\end{lemma}

The expander code $\mathcal{C}$ corresponding to the graph $G$ is defined as the set of all binary vectors $x \in \mathbb{F}_2^N$ such that $xH = 0$, where $H$ is the parity-check matrix based on the adjacency matrix of the bipartite graph $G$. Expander codes are known for their good distance and efficient decoding algorithms, as demonstrated in numerous previous studies, such as \cite{sipser1996expander, feldman2006lp, viderman2013linear, chen2023improved}.

\begin{lemma}[\cite{viderman2013linear}]
If $G$ is an $(K, (1-\varepsilon)D)$-expander for some $\varepsilon < 1/3$ and $\mathcal{C}$ is a code defined by it, then $\mathcal{C}$ is decodable in linear time from $(1 - \frac{\varepsilon}{1-2\varepsilon}) K$ errors.
\end{lemma}

Combining these two results, we choose $\varepsilon = 1/4$ and $\beta = 1/2$. This graph $G$ has $|R(G)|/|L(G)|=1/2$ and is an $(\alpha N, (1-\varepsilon)D)$-expander with $\alpha = 1/50000$ and $\varepsilon = 1/4$, which implies the following codes.

\begin{corollary}\label{lemma_mother_d}
There exists an explicit family of binary codes with a rate of $1/2$ that is uniquely decodable from $\frac{1}{100,000}$ fraction of errors in linear time.
\end{corollary}

We present our first theorem, which gives the unique decoding property of the graph-concatenated code $G(\mathcal{C})$ based on the properties of $\mathcal{C}$ and $G$.

\begin{theorem}\label{theorem_distance}
    Given that $ \mathcal{C} $ is a code of length $ M $ over the alphabet $\Sigma$, which has relative distance $ \delta $, and can be uniquely decoded from up to $ \delta/2 $ fraction of errors in time $ T $, and $G = (L \sqcup R, E)$, with $|L| = N$, $|R| = M$, and left degree $D$, is a $ (K, \delta/2) $-disperser, then $ G(\mathcal{C}) $ has relative distance at least $ (1- K/N) $ and can be uniquely decoded from $e$ errors and $s$ erasures as long as $2e+s<(N-K)$ in time $O(ND\log(|\Sigma|))+T$.
\end{theorem}

\begin{proof}
    We first establish the distance of $ G(\mathcal{C}) $. The unique decoding algorithm and its proofs will be provided later. For any two different codewords $ c^1, c^2 \in \mathcal{C} $, let $ \mathcal{S} $ denote the subset of $ L(G) $ where $ G(c^1) $ and $ G(c^2) $ are identical. We claim that $ |\mathcal{S}| < K $; otherwise, since $ G $ is a $ (K, \delta/2) $-disperser, $ |\Gamma(\mathcal{S})| \geq (1 - \delta/2) M $. By our construction, $ c^1 $ and $ c^2 $ take the same value for each position in $ \Gamma(\mathcal{S}) $, and thus $d(c^1,c^2)\leq \frac{\delta M}{2}$, which contradicts the condition that the relative distance of $ \mathcal{C} $ is at least $ \delta $.
\end{proof}

In our proof, we allowed an extra $\frac{1}{2}\delta$ room for the disperser's error. While this might make the distance a bit loose, it is necessary for the subsequent unique decoding algorithm. By employing a specific code which is uniquely decodable in linear time stated in Corollary \ref{lemma_mother_d}, we get the following results.

\begin{corollary}\label{corollary_distance}
    For any $\varepsilon > 0$, there exists an explicit family of $\mathbb{F}_2$-linear codes over an alphabet of size $\quasipoly\left(1/\varepsilon\right)$, which have rate $\Omega(\varepsilon)$ such that a code with block length $n$ in the family can be decoded from $e$ errors and $s$ erasures as long as $2\cdot e + s < (1-\varepsilon)n$ in time $O(n \cdot\polylog(1/\varepsilon))$.
\end{corollary}

\begin{proof}
We use the code described in Lemma \ref{lemma_mother_d} as $\mathcal{C}$. We then select a disperser as shown in \ref{theorem_disp} by setting $\delta = \frac{1}{50,000}$ and $K = \varepsilon N$. The alphabet size of $G(\mathcal{C})$ is computed by $2^D$, which is $\quasipoly(1/\varepsilon)$, and the rate of $G(\mathcal{C})$ is given by $\frac{R(\mathcal{C})N}{MD}=\Omega(\varepsilon)$. Since the expander code $\mathcal{C}$ is a linear binary code, $G(\mathcal{C})$ inherits its $\mathbb{F}_2$ linearity.
\end{proof}

If we replace the disperser given in Theorem \ref{theorem_disp} with an disperser with optimal degree given by Lemma \ref{lemma_disp}, we can reduce the alphabet size of $G(\mathcal{C})$ as follows.

\begin{corollary}\label{corollary_distance}
    For any $\varepsilon > 0$, there exists a family of $\mathbb{F}_2$-linear codes over an alphabet of size $\poly(1/\varepsilon)$, which have rate $\Omega(\varepsilon)$ such that a code with block length $n$ in the family can be decoded from $e$ errors and $s$ erasures as long as $2\cdot e + s < (1-\varepsilon)n$ in time $O(n \cdot \polylog(1/\varepsilon))$, and can be constructed in probabalistic time $\poly(n,\log(1/\varepsilon))$ with success probability at least $1-\exp(-\log(1/\varepsilon)\cdot n)$.
\end{corollary}

\subsection*{Linear-Time Unique Decoding}

We will present a linear-time decoding algorithm for such a code, based on the decoding algorithm for the mother code $\mathcal{C}$.

Given a received word $ y: L(G) \mapsto \Sigma^D $, we distribute all values from the left side to the right side via the graph. We then execute the erasure process to turn some `errors' into `erasures' by identifying inconsistencies on the right side.

If two vertices on the left, $ u $ and $ v $, satisfy $ y(u), y(v) \neq \perp $ and $\Gamma_i(u) = \Gamma_j(v)$ for some $ i, j \in [D] $, but $ y(u)_i \neq y(v)_j $, then we replace the values at $ u $ and $ v $ with $\perp$ (erasures). This process repeats until no such pair $ u, v $ exists. Subsequently, for each vertex $ w \in R(G) $, there are two possible scenarios: either all neighboring values of $ w $ are erased, or $ w $ receives one or more values from the left, all of which are identical. In the first case, we assign an arbitrary value to $ w $; in the latter case, $ w $ is assigned that value. This yields a word $z\in \Sigma^R$ on the right side. We then apply a unique decoding algorithm to obtain the correct codeword $ c $ and encode it using the graph to get $ G(c) $. The algorithm is provided in Algorithm \ref{algorithm_1}.

\begin{algorithm}[h]
    \caption{Unique Decoder for $ G(\mathcal{C}) $}
    \label{algorithm_1}
    \SetAlgoLined
    \KwIn{Received word $ y: L \mapsto \Sigma^D $}
    \KwOut{A codeword $ G(c) \in G(\mathcal{C}) $} 
    \ForEach{$ w \in R(G) $}{
        \While{$ \exists u, v \in L(G) $ with $ y(u), y(v) \neq \perp $ and $ i, j \in [D] $ such that $ w = \Gamma_i(u) = \Gamma_j(v) $, $ y(u)_i \neq y(v)_j $}{
            $ y(u) \leftarrow \perp $\;
            $ y(v) \leftarrow \perp $\;
        }
        \If{$ \exists u \in L(G) $ and $ i \in [D] $ such that $ v = \Gamma_i(u) $ and $ y(u) \neq \perp $}{
            $ z(v) \leftarrow y(u)_i $\;
        }
        \Else{
            $ z(v) $ is assigned an arbitrary value in $\Sigma$\;
        }
    }
    Perform unique decoding algorithm on $ z $ to obtain $ \bar{z} \in \mathcal{C} $\;    
    \If{The half distance between $ G(\bar{z}) $ and $ y $ is less than $N-K$}{
        \Return $ G(\bar{z}) $\;
    }
    \Else{
        \Return Failed\;
    }
\end{algorithm}

\begin{lemma}Suppose $y$ has $s$ erasures and $e$ different (non-erasure) values from some codeword $G(c) \in G(\mathcal{C})$. Then, as long as $2e + s < N-K$, this algorithm will always return the correct codeword $G(c)$.
\end{lemma}

\begin{proof}
During the while loop, define $E \subseteq L$ as the subset of left vertices where errors (values that differ from $G(c)$) occur for $y$, and $S \subseteq L$ as the subset where erasures occur for $y$. Let $Y = L \setminus (E \cup S)$ be the vertices with correct symbols. 

For each execution of the while loop and each pair $(u, v)$ (they could be the same), if they are both not erased and $\exists i, j \in [D]$ such that $\Gamma_i(u) = \Gamma_j(v)$ but $y(u)_i \neq y(v)_j$, then at least one of $u$ or $v$ is from $E$. Otherwise, if $u$ and $v$ are both from $Y$, then we would have $y(u)_i = y(v)_j = c(w)$ for $w = \Gamma_i(u)$, contradicting the assumption that $y(u)_i \neq y(v)_j$.

This indicates that after each execution of the while loop, the size of $S$ increases by at most $2$, and the size of $E$ decreases by at least $1$ (if $u = v$, then the size of $S$ increases by $1$, and the size of $E$ decreases by $1$). Since before the algorithm $2|E| + |S| < N-K$, this condition always holds during the algorithm, ensuring that $|Y| \geq K$. Suppose $\bar{Y}$ is the set of left vertices with consistent values with $G(c)$ after exiting the while loop, which is also of size at least $K$. We claim that for each $w \in \Gamma(\bar{Y})$, $w$ will always receive one correct symbol, since for each $v \in \bar{Y}$, it always stays in $Y$ and is never erased. Considering the disperser property, we infer $|\Gamma(\bar{Y})| \geq (1 - \delta/2)M$, which tells us that $z$ has at least $(1 - \delta/2)M$ positions consistent with $c$. Then, the decoding algorithm for $\mathcal{C}$ will always return $c$ given $z$.
\end{proof}

Note that the distribution from left to right takes time  $O(N D \log|\Sigma|)$, and checking the condition in the while loop only requires scanning all received values and performing at most $\frac{1}{2} N$ erasure processes.  Meanwhile, for each value erased at $v \in L(G)$, we only need to modify $D$ symbols on the right side. Therefore, the entire algorithm can be completed in time $O(N D \log|\Sigma|) +T$.

\section{Polynomial-Time List Decodable Codes}\label{sec_list}

In this section, we will deal with list-decodable codes. Similar to the construction of uniquely decodable codes, we will use the disperser provided by Theorem \ref{theorem_disp}. However, the mother code we employ is stronger than that used for uniquely decodable codes, as we require a list-recoverable code, utilizing a construction from \cite{hemenway2019local}.

\begin{lemma}\cite{hemenway2019local}\label{lemma_recover}
For any $\ell$, there exists a prime power $q = \ell^{O(1)}$, and a family of deterministically polynomial-time constructible $\mathbb{F}_q$-linear codes $\mathcal{C}$ of rate $\frac{1}{2}$ and relative distance $\delta = \Omega(1)$. A code of length $n$ in this family is $(\rho, \ell, L)$-list recoverable in time $\poly_{q,\ell}(n)$, where $\rho=\Omega(1)$ and $L = \exp\left(\exp\left(\exp\left(\log^{\ast} n\right)\right)\right)$.
\end{lemma}

The following theorem provide a bridge from a list recoverable code $\mathcal{C}$ to a list decodable code $G(\mathcal{C})$ using a disperser $G$.

\begin{theorem}\label{theorem_list1}
Given that $\mathcal{C}$ is a code of length $M$ over the alphabet $\Sigma$, which can be $(\rho,\ell,\bar{L})$-list recoverable in time $T$, and $G : (L \sqcup R, E)$, with $|L| = N$, $|R| = M$, and left degree $D$, is a $(K, \rho)$-disperser, then $G(\mathcal{C})$ is $(1-\frac{1}{\ell}-K/N, \bar{L})$-list decodable in time $O(N D \log|\Sigma|) +T$.
\end{theorem}

We will leave the list decoding algorithm and its proof for $G(\mathcal{C})$ to the next paragraph. Here, we present the specific results derived from this theorem. We take $\mathcal{C}$ as the code described in Lemma \ref{lemma_recover} by setting $\ell = \frac{2}{\varepsilon}$. The dispersers are taken from Theorem \ref{theorem_disp} for the explicit version and Lemma \ref{lemma_disp} for the semi-explicit version, with $K = \frac{\varepsilon N}{2}$ and error $\delta=\rho$ as specified in the code. We have the following corollary.

\begin{corollary}\label{corollary_list1}
For any $\varepsilon > 0$, there exists an explicit family of codes over an alphabet of size $\quasipoly(1/\eps)$ which have rate $\Omega(\varepsilon)$ such that a code with block length $n$ in this family can be list decoded up to $(1-\varepsilon)$ fraction of errors with list size $L = \exp\exp\exp_{\eps}(\log^\ast n)$ in time $\poly_{\varepsilon}(n)$. Moreover, there exists an family of codes over an alphabet of size $\poly(1/\varepsilon)$ that maintains the same rate and list-decoding properties, and a code with block length $n$ in this family can be constructed in probabalistic time $\poly(n,\log(1/\varepsilon))$ with success probability at least $1-\exp(-\log(1/\varepsilon)\cdot n)$.
\end{corollary}

\subsection*{Polynomial-Time List Decoding Algorithm} The algorithm in unique decoding is no longer valid because the number of incorrect left vertices often exceeds the number of correct vertices. However, the condition that $\mathcal{C}$ is list recoverable also provides us with one advantage: we do not need to reduce the number of candidates for the value of each right vertex to $1$; we only need to keep them within $\ell$. Therefore, if the number of candidates for some $w\in R(G)$ exceeds $\ell$, we can find no more than $\ell$ candidates to remove. To ensure that enough correct left vertices remain, we need to ensure that for each correct symbol erased, at least $(\ell-1)$ incorrect symbols are also erased. To achieve this, we need to find $\ell$ distinct values sent to a single vertex on the right. For each vertex $w\in R(G)$, if there are $\ell$ different values corresponding to it, we are required to erase them. To be specific, if there are $\ell$ left vertex $u_1,\ldots,u_{\ell}$ such that there are $\ell$ indices $t_i,\ldots,t_{\ell}$, with $w = \Gamma_{t_i}(u_{i})$ for all $i$, but $y(u_i)_{t_i}$ are all different. Then we will replace the value of each left vertex $y(u_i)$ with $\perp$. However, we need to be careful for the cases where one left vertex has repeated neighbors on the right. In this case, the erased incorrect vertices could be much smaller. Thus, before taking the above step, we will find all cases where $\Gamma_{i}(u) = \Gamma_{j}(u)$, but $y(u)_i\neq y(u)_j$ for some $i\neq j$, and erase these $u$. The algorithm is provided in Algorithm \ref{algorithm_2}.

\begin{algorithm}[h]\label{algorithm_2}
    \caption{List Decoder for $G(\mathcal{C})$}
    \SetAlgoLined
    \KwIn{Received word $ y: L \mapsto \Sigma^D $}
    \KwOut{A Set $\mathcal{S}$ containing all codeword $G(c) \in G(\mathcal{C}) $ with $d(G(c),y)\leq (1-\gamma)N$, where $\gamma = \frac{1}{\ell}+K/N$}
    \ForEach{$ w \in R(G) $}{
        \While{$\exists u\in L(G)$, with $y(u)\neq \perp$ and $i\neq j\in [D]$ such that $w = \Gamma_{i}(u) = \Gamma_{j}(u)$ but $y(u)_i\neq y(u)_j$}{
            $y(u)\leftarrow \perp$
        }
        \While{$ \exists u_1,\ldots, u_{\ell} \in L $ with $ y(u_t)\neq \perp $ for all $t\in [\ell]$ and $ i_1,\ldots,i_{\ell} \in [D] $ such that $ w = \Gamma_{i_t}(u_t)$, and $ y(u_t)_{i_{t}} \neq y(u_k)_{i_{k}}$ for all $t\neq k\in[\ell]$}{
            $ y(u_t) \leftarrow \perp$ for all $t\in [\ell]$\;
        }
        Set $\mathcal{T}(w) = \{y(u)_i:  \Gamma_i(u) = w\}$ \;
    }
    Perform list recovery algorithm on $\mathcal{T}(w)$'s to obtain a list $\mathcal{S}^{\prime}$, with size at most $\bar{L}$\;
    \ForEach{$z\in \mathcal{S}^{\prime}$}{    
    \If{$ d(y,G(z))) < 1 - \gamma$}{
        $\mathcal{S} = \mathcal{S} \cup G(z) $\;
    }}
    \Return $\mathcal{S}$\;
\end{algorithm}

Similar to unique decoding, we will prove the following lemma:

\begin{lemma}
Denote $\gamma = \frac{1}{\ell}+K/N$. Suppose $y$ has distance at most $(1 - \gamma)N$ from some codeword $G(c) \in G(\mathcal{C})$. Then $G(c)$ will be contained in $\mathcal{S}$, and $|\mathcal{S}|\leq \bar{L}$.
\end{lemma}

\begin{proof}
During the while loop, define $E \subseteq L$ as the subset of left vertices where errors (values that differ from $G(c)$) occur for $y$, and $S \subseteq L$ as the subset where erasures occur for $y$. Let $Y = L \setminus (E \cup S)$ be the vertices with correct symbols. 

\begin{claim}
If $u$ enters the first while loop, then $u \in E$. After the first while loop, $E$ will decrease by 1 and $S$ will increase by 1.
\end{claim}

\begin{proof}
If $u \in Y$, then $c(w) = y(u)_i = y(u)_j$, which leads to a contradiction. Since $u \notin S$, we know that $u \in E$. After the while loop, $u$ will move into $S$.
\end{proof}

\begin{claim}
If $u_1, \ldots, u_\ell$ enter the second while loop, then at most one of them is in $Y$. After the second while loop, $E$ will decrease by at least $\ell - 1$ and $S$ will increase by $\ell$.
\end{claim}

\begin{proof}
After the first while loop, we know that these $u_i$'s are all different, and among them, at least $\ell - 1$ of them belong to $E$. Otherwise, suppose $u_t \neq u_k$ for some $t \neq k \in [\ell]$ and they are both in $Y$, then $y(u_t)_{i_t} = y(u_k)_{i_k} = c(w)$, which leads to a contradiction.
\end{proof}

Suppose in the whole algorithm the first while loop takes $w_1$ iterations and the second while loop takes $w_2$ iterations. It is easy to see that $w_2 \leq \frac{N}{\ell}$. The final size of $S$ will be $w_1 + \ell w_2$, which is at most $N$. Initially, $|E| \leq (1 - \gamma)N$, which is decreased by $w_1 + (\ell - 1)w_2$ during the algorithm. Therefore, after the first for loop, we have:
\begin{equation}
\begin{aligned}
|E| + |S| \leq & (1 - \gamma) N - w_1 - (\ell - 1)w_2 + w_1 + \ell w_2 \\
=& (1 - \gamma)N + w_2 \\
\leq& (1 - \gamma) N + \frac{N}{\ell} \\
=& (1 - \gamma + \frac{1}{\ell}) N \\
=& (1 - \frac{K}{N}) N.
\end{aligned}
\end{equation}
Thus, $|Y| \geq K$ finally.

In the algorithm, the while loop will continue until no $w \in R(G)$ receives more than $\ell$ different symbols from left vertices. If some $v \in L(G)$ remains in $Y$ finally, it must have been in $Y$ initially, since $Y$ is never enlarged during this process, and the value on $v$ is exactly $G(c)(v)$. Therefore, for any $w \in R(G)$, $|\mathcal{T}(w)| \leq \ell$, and if $w \in \Gamma(Y)$, then $c(w) \in \mathcal{T}(w)$. Considering the disperser property, we infer $|\Gamma(Y)| \geq (1 - \delta)M$. Thus, for at least $1 - \delta$ fraction of the vertices $w$, the value of $c(w)$ will be in $\mathcal{T}(w)$. By the list decoding property of $\mathcal{C}$, this algorithm will always produce a list $\mathcal{S}^{\prime}$ with size at most $\bar{L}$ containing $c$. Therefore, $|\mathcal{S}|\leq \bar{L}$ and $G(c) \in \mathcal{S}$.
\end{proof}

The distribution and erasure process for list decoding is similar to unique decoding, and thus we requires $O(ND\log(\Sigma|))+T$ time.
\section{List Decodable Code with Near-Optimal List Size}\label{sec_ran}
The above analysis indicates that the list decodability and decoding algorithm of the concatenated code rely on the list-recoverability of the mother code. In this section, we will present a method to construct a list decodable code with near-optimal list size using the same framework but with a different analysis method. Specifically, we explore average-radius list decoding, a stronger condition compared to combinatorially list decoding. The mother code we start with does not necessarily need to have the list-recoverable condition. Instead, we will demonstrate that a code with constant distance and constant rate is already sufficient. We present the definition of average-radius list decoding again to emphasize it.

\begin{definition}\label{definition_average}
A code $\mathcal{C}$ is said to be $(\rho, L)$-average-radius list decodable if for every $y \in \Sigma^n$ and all different codewords $c^1, \ldots, c^{L+1} \in \mathcal{C}$, $\sum_{i=1}^{L+1} \delta(c^i, y) > (L+1) \rho$.
\end{definition}

The benefit of adopting average-radius list decoding lies in its ability to simplify the problem by converting it into a linear format, which is well analyzed in \cite{rudra2014every}. Specifically, we will provide the following definitions and propositions to clarify this method.

\begin{definition}\label{definition_plurality}
For a set $\Lambda \subseteq \Sigma^n$, let $\mathbf{pl}_j(\Lambda)$ denote the plurality of index $j \in [n]$:
\begin{equation}
\mathbf{pl}_j(\Lambda) =  \max_{\alpha \in \Sigma} \left| \{ c \in \Lambda : c_j = \alpha \} \right|.
\end{equation}
\end{definition}

\begin{proposition}\label{proposition_plu_ave}
Consider a code $\mathcal{C}$ over $\Sigma$ with length $n$. If, for any set of $L$ different codewords $\Lambda$, the inequality
$$
\sum_{j \in [n]} \mathbf{pl}_j(\Lambda) < \varepsilon n L
$$    
holds, then $\mathcal{C}$ is $(1-\varepsilon, L-1)$-average-radius list decodable.
\end{proposition}

The proposition presented above demonstrates that by examining the plurality for each set of $L$ independent codewords within a code, we can infer the list decodability of the code. In the following analysis, to demonstrated that a code is list decodable with list size $L-1$, we will demonstrated that for each $L$ codewords $\Lambda\subseteq\mathcal{C}$, $\sum_{j\in [n]}\mathbf{pl}_j(\Lambda)$ is bounded. It is not hard to prove by probabilistic argument that a random code with length $n$ and rate $\Omega(\varepsilon)$ will, with high probability, satisfy the condition that for any $L = O(1/\varepsilon)$ different codewords, 
$$
\sum_{j \in [n]} \mathbf{pl}_j(\Lambda) < \varepsilon n L.
$$

To illustrate how this can be applied, we begin with a preliminary example to explain the analysis method through the following theorem. Here, we directly use the conclusion provided by Theorem \ref{theorem_distance}, proving that it is list decodable up to the optimal list size.

\begin{theorem}\label{theorem_list2}
    For any small $\varepsilon > 0$, there exists an explicit family of codes over an alphabet of size $\mathrm{quasipoly}\left(\frac{1}{\varepsilon}\right)$, which have rate $\Omega(\varepsilon^2)$. Moreover, a code in this family of length $n$ is $(1-\varepsilon, L )$-list decodable in time $\poly_{\varepsilon}(n)$, where $L = O(\frac{1}{\varepsilon})$.
\end{theorem}
\begin{proof}
    In this proof, we treat $N=n$ for the sake of convenience. Note that Theorem \ref{theorem_distance} tells us that we can start from a code with constant rate and constant distance and, by applying a $(K,\delta/2)$-disperser with $K=\varepsilon^2N/2$, and $\delta$ corresponding to the distance of the mother code and the resulted code has distance $1-\varepsilon^2/2$, and rate $\Omega(\varepsilon^2)$. And we will claim that this code is actually $(1-\varepsilon,L-1)$-list decodable for $L=\lfloor \frac{2}{\varepsilon} \rfloor$.
    The proof we present here is similar to the proof of the Johnson bound. We use a double counting method to determine number of occurrences where $G(c^i)(v)= G(c^j)(v)$ for different $c^i \neq c^j \in \Lambda$ and some $v\in L(G)$. Since the distance of code is $1-\varepsilon^2/2$, fixing two codewords $c^{i}\neq c^{j}$, we know that the number of such $v$ is at most $\varepsilon^2/2 N$. Fixing one vertex $v$, the number of code pairs $G(c^i)\neq G(c^j)$ which are the same values on $v$ is at least ${\mathbf{pl}_v\choose 2}+1$, thus, we have
\begin{equation}
\begin{aligned}
    \sum_{v\in L(G)}\mathbf{pl}_v(\Lambda)&\leq N+ \sum_{v\in L(G)}{\mathbf{pl}_v(\Lambda)\choose 2}\\ &\leq N + {L\choose 2}\cdot \frac{1}{2}\varepsilon^2 N \\& < \varepsilon L N.
\end{aligned}
\end{equation}
Therefore, this code is $(1-\varepsilon,\lfloor\frac{2}{\varepsilon}\rfloor-1)$-list decodable.

To construct a list decoding algorithm for this code, we employ the code provided by Lemma \ref{lemma_recover} as the mother code $\mathcal{C}$. In this case, the our disperser needs $K=\varepsilon N$. By setting the input size $\ell$ as $\lceil\frac{2}{\varepsilon}\rceil$ for the code $\mathcal{C}$, the error $\delta$ of the disperser is set to be the half of the distance, or recovery radius, of the mother code $\mathcal{C}$, whichever is smaller. Applying Algorithm \ref{algorithm_2}, and the final list size will be bounded by $\lfloor\frac{2}{\varepsilon}\rfloor-1$.
\end{proof}

Next, we will explain that in our graph-concatenation model, the plurality of the code can also be bounded if the rate is near optimal, say $\Omega(\varepsilon / \poly\log(1/\varepsilon))$. We believe that the rate can actually be $\Omega(\varepsilon)$, which we leave as an open question. However, we present a result that already beats the Johnson bound. Our graphs are based on what we call \textit{multi-set dispersers} with optimal parameters, which are not fully explicit but can be randomly constructed. In other words, we assert that most graphs will work for our construction.

Before introducing the objects we used, we will provide a counting result that will be used in the later proof.

\begin{lemma}\label{lemma_count}
For any set $ \mathcal{S} $ of size $ L $ and arbitrary $ N $ subsets $ \mathcal{S}_1, \ldots, \mathcal{S}_N $ of $ \mathcal{S} $, each of size $ \beta L $ (where $ \beta \leq \frac{1}{2} $ and $ \beta L\geq 2 $), there exists a subset $ \mathcal{T} $ of $ \mathcal{S} $ of size $ \lceil \frac{2}{\beta} \rceil $, such that at least $ 0.3 N $ subsets in $\mathcal{S}_1, \ldots, \mathcal{S}_N $ have an intersection with $ \mathcal{T} $ of size at least $2$.
\end{lemma}

\begin{proof}
We randomly select a subset $ \mathcal{T} $ of $ \mathcal{S} $ with size $ \lceil \frac{2}{\beta} \rceil $, and compute 
\begin{equation}
\mathbb{E}(|\mathcal{S}_i: |\mathcal{S}_i\cap \mathcal{T} |\geq 2 |) =\sum_{i\in[N]}\mathrm{Pr}(|\mathcal{S}_i\cap \mathcal{T} |\geq 2).
\end{equation}
For each subset $ \mathcal{S}_i $, the probability that the size of the intersection $ |\mathcal{S}_i \cap \mathcal{T}| $ is less than 2 is at most:

\begin{equation}
(1 - \beta)^{\lceil\frac{2}{\beta}\rceil} + \lceil \frac{2}{\beta}\rceil\cdot \beta \left(1 - \beta \right)^{\lceil\frac{2}{\beta}\rceil-1} \leq \frac{1}{e^2} + \frac{4}{e^2} < 0.7.
\end{equation}

Therefore, the probability that $ |\mathcal{S}_i \cap \mathcal{T}| \geq 2 $ is greater than $ 0.3 $ for each $i\in [N]$, which implies that
\begin{equation}
    \mathbb{E}(|\mathcal{S}_i: |\mathcal{S}_i\cap \mathcal{T} |\geq 2 |) > 0.3N.
\end{equation}
Thus, such a subset $ \mathcal{T} $ exists, completing the proof.
\end{proof}

We will introduce a completely new concept called a \textit{multi-set disperser}s here, which will be used in later code constructions.

\begin{definition}\label{definiton_muldisp}
A $(K,\delta)$-multi-set disperser is a bipartite graph $G: (L\sqcup R, E)$ with a uniform left degree denoted as $D$, which has property that for every $t \geq 2$ and $t$ different subsets $\mathcal{S}_1,\mathcal{S}_2 \ldots, \mathcal{S}_t$ of $R$ with $|\mathcal{S}_i \Delta \mathcal{S}_j| \geq \delta\cdot |R|$ for each $i\neq j\in[t]$, the number of left vertices $v \in L(G)$ such that $\Gamma(v)$ is contained in some $(\mathcal{S}_i \Delta \mathcal{S}_j)^c$ is at most $t \cdot K$.
\end{definition}

A multi-set disperser is a generalization of a disperser, by considering several symmetric differences of sets on the right. However, there is also a conversion relationship between these two objects. Specifically, a $(K, \delta)$-multi-set disperser is a $(2K, \delta)$-disperser. Conversely, a $(K, \delta)$-disperser is a $(\sqrt{\frac{KN}{2}}, \delta)$-multi-set disperser.

\begin{lemma}\label{lemma_rand_mul}
For every integers $N>0, 0< K < N/2$, and $0<\delta<1$, there exists an $(K, \delta)$-multi-set dispersers $G: (L\sqcup R, E)$, where $|L|=N$, $|R|=M$, with left degree $D$ and
\begin{itemize}
    \item $M = \lfloor K/N\log(N/K)\rfloor\cdot N$
    \item $D = \lceil 4\log(N/K)/\delta \rceil$.  
\end{itemize}
Moreover, a random bipartite graph with $N$ left vertices, $M$ right vertices, and left-degree $D$ as specified above above is a $(K, \delta)$-multi-set disperser with probability at least $1 - \exp(1 - \log(N/K)N)$.
\end{lemma}

\begin{proof}
We prove this by a probabilistic argument. Denote $K/N=\alpha$. Take a random bipartite graph. By definition, we only need to discuss the case where $t < \frac{1}{\alpha}$. Fix such a $t$ and $t$ different subsets $\mathcal{S}_1, \mathcal{S}_2, \ldots, \mathcal{S}_t$ of $R$ with $|\mathcal{S}_i \Delta \mathcal{S}_j| \geq \delta \cdot M$ for each $i\neq j \in [t]$. Let $\mathcal{T}$ be the set of all vertices $v \in L$ such that the support of $\Gamma(v)$ is contained in some $(\mathcal{S}_i \Delta \mathcal{S}_j)^{\mathrm{c}}$. For each $v \in L$,
\begin{equation}
\begin{aligned}
    \mathrm{Pr}(v\in \mathcal{T}) \leq& {t\choose 2}(1-\delta)^D \\ \leq& t^2(1-\delta)^{\frac{4}{\delta}\log(1/\alpha)} \\\leq& t^2 e^{-4\log(1/\alpha)} \\\leq& t \alpha^{4\log(e)-1} \leq t \alpha^{4}.
\end{aligned}
\end{equation}
By the Chernoff bound, the probability that $|\mathcal{T}| > t K$ is at most
\begin{equation}
\left(e \cdot \frac{t \alpha^4}{t \alpha}\right)^{t K} = (e \alpha^3)^{tK}.
\end{equation}
By the union bound for all $t<\frac{1}{\alpha}$ and all possible different $t$ subsets $\mathcal{S}_1,\mathcal{S}_2,\ldots,\mathcal{S}_t$, the probability that this graph is not a $(K,\delta)$-multi-set disperser is at most 
\begin{equation}
\begin{aligned}
&\sum_{t=1}^{\lfloor\frac{1}{\alpha}\rfloor} 2^{tM} (e \alpha^3)^{t K} \\
\leq & \frac{1}{\alpha} \cdot \alpha^{-t\alpha N} \cdot (e \alpha^3)^{t K} \\
\leq & \alpha^{N(2 - \log(e)) - 1} \\
\leq & \alpha^{\frac{N}{2} - 1},
\end{aligned}
\end{equation}
which is exponentially decreasing. Therefore, such multi-set dispersers exist.
\end{proof}

Next, we present the main theorem: we can use multi-set dispersers as graphs for concatenation, resulting in codes with improved list decoding properties.

\begin{theorem}\label{theorem_last}
    Given that $\mathcal{C}$ is a binary code of relative distance $\delta$, and $G: (L \sqcup R, E)$, with $|L| = N$, $|R| = M$, and left degree $D$, is a $(K,\delta)$-multi-set disperser, then $G(\mathcal{C})$ is $(1-10\frac{K}{N}(\ln(\bar{L})+2), \bar{L}-1)$-list decodable for $\bar{L} = \lceil \frac{N}{5K} \rceil$.
\end{theorem}

Before proving this theorem, we will introduce additional notations.

\begin{definition}\label{definiton_plu_cond}
Given $0<\beta,\delta <1$, and an integer $L$, a code $\mathcal{C}$ of length $n$ is said to satisfy the $(\beta, \delta, L)$-plurality condition if, for any set of $L$ different codewords $\Lambda\subseteq \mathcal{C}$, the condition
$$
|\{j \in [n] : \mathbf{pl}_j(\Lambda) \geq \beta L\}| \leq \delta n
$$
holds.
\end{definition}

\begin{lemma}\label{lemma_abel}
Given any $0 < \varepsilon < 1$, if $\mathcal{C}$ satisfies the $(\beta, \varepsilon/\beta, L)$-plurality condition for all $\varepsilon < \beta \leq \frac{1}{2}$, then for any set of $L$ different codewords $\Lambda\subseteq \mathcal{C}$,
\begin{equation}
\sum_{j \in [n]} \mathbf{pl}_j(\Lambda) \leq \varepsilon L n (\ln(L) + 2).
\end{equation}
\end{lemma}

\begin{proof}
By Abel summation, the left side of the inequality is
\begin{equation}
\begin{aligned}
\sum_{j \in [n]} \mathbf{pl}_j(\Lambda)
= &\sum_{j:\mathrm{pl}_j(\Lambda) \leq \frac{1}{2}L} \mathbf{pl}_j(\Lambda)
+\sum_{j:\mathrm{pl}_j(\Lambda) \geq \frac{1}{2}L} \mathbf{pl}_j(\Lambda)
\\\leq & \sum_{i=1}^{\lfloor L/2\rfloor} \left\lfloor \frac{\varepsilon L n}{i} \right\rfloor + 2 \varepsilon L n \\\leq & \varepsilon L n (\ln(L) + 2).
\end{aligned}
\end{equation}
\end{proof}

\begin{proof}[Proof of Theorem \ref{theorem_last}] For $\bar{L}$ different codewords $c^1, c^2, \ldots, c^{\bar{L}}$ valued on the right vertices $R(G)$, we use $\mathcal{S}_i$ to denote the subset of $R(G)$ where $c^i$ takes value $1$. Therefore, $\mathcal{S}_i \Delta \mathcal{S}_j$ is exactly the subset of right vertices where $c^i$ and $c^j$ have different values, and thus $|\mathcal{S}_i \Delta \mathcal{S}_j|\geq \delta M$.

For any $\beta > \frac{10K}{N}$, which is also greater than $2/\bar{L}$, let $\mathcal{Q}_{\beta}$ be the set of all vertices $v$ with $\mathbf{pl}_v(\Lambda) > \beta \bar{L}$. We denote $\mathrm{maj}(v)$ as the symbol that appears most frequently on $v$ for these $L$ codewords and denote $\mathcal{T}(v)$ as all codewords whose value on $v$ is $\mathrm{maj}(v)$. By Lemma \ref{lemma_count}, we know that there is a subset $\mathcal{Q}_{\beta}^{\prime} \subseteq \mathcal{Q}_{\beta}$ with size at least $0.3 |\mathcal{Q}_{\beta}|$ and a subset $\mathcal{T} \subseteq \Lambda$ of size $\lceil \frac{2}{\beta} \rceil$, such that $|\mathcal{T}(v) \cap \mathcal{T}| \geq 2$ for all $v \in \mathcal{Q}_{\beta}^{\prime}$.

The condition $|\mathcal{T}(v) \cap \mathcal{T}| \geq 2$ implies that there are two codewords $c^i \neq c^j$ from $\mathcal{T}$ such that $G(c^{i})$ and $G(c^{j})$ take the same value on $v$, which gives $\Gamma(v) \subseteq (\mathcal{S}_i \Delta \mathcal{S}_j)^c$. By the multi-set disperser property, this gives that $|\mathcal{Q}_{\beta}^{\prime}| \leq \lceil \frac{2}{\beta} \rceil K$.

\begin{equation}
|\mathcal{Q}_{\beta}| < \frac{10}{3}|\mathcal{Q}_{\beta}'| \leq \frac{10}{3} \cdot \lceil \frac{2}{\beta} \rceil K \leq \frac{10}{\beta} K.
\end{equation}

Thus, this code satisfies the $(\beta, \frac{10K}{\beta N}, \bar{L})$-plurality condition for all $\frac{10K}{N} < \beta \leq \frac{1}{2}$. By Lemma \ref{lemma_abel}, we get that this code satisfies that for any $\bar{L}$ codewords,

\begin{equation}
\sum_{j \in [N]} \mathbf{pl}_j(\Lambda) < 10K \bar{L} (\ln(\bar{L}) + 2).
\end{equation}

Therefore, it is $(1-10\frac{K}{N} (\ln(\bar{L}) + 2), \bar{L}-1)$-list decodable for $\bar{L} = \lceil \frac{N}{5K} \rceil$.
\end{proof}

In our context, the mother code must be binary due to the definition of the multi-set disperser. Specifically, we consider the symmetric difference of two subsets of the right vertices, so $\mathcal{S}_i$ corresponds to the bits where $c^i$ has a value of $1$. Consequently, we cannot employ a list recovery algorithm. To address this issue, we use the "folding" operation to combine a code with an efficient list decoding algorithm and a code that is list-decodable with a small list size. This approach allows us to achieve a polynomial-time list decoding algorithm for the above codes.

Given any $\varepsilon > 0$, let $\mathcal{C}_1$ be $(\rho,\ell,L)$-list recoverable code given by Lemma \ref{lemma_recover} with length $M$ by setting the input list size $\ell = \frac{2}{\varepsilon}$. This code is over alphabet $\mathbb{F}_q$ for some prime power $q = \ell^{O(1)}$, and it has constant relative distance $\delta$. Let $\mathcal{C}_2$ be a binary code with size $|\mathcal{C}_2| = q$, constant rate, and constant distance $\delta'$ (the code length of $\mathcal{C}_2$ is $\Theta(q) = \Theta(\log(1/\varepsilon))$). Such a code exists and can be found by brute force. Denote $\mathcal{C}_1 \circ \mathcal{C}_2$ as the concatenated code by using $\mathcal{C}_1$ as the outer code and $\mathcal{C}_2$ as the inner code. It is easy to compute that $\mathcal{C}_1 \circ \mathcal{C}_2$ is a binary code with constant rate and constant relative distance $\delta \delta'$, and the length of this code is $M' = \Theta(\log(1/\varepsilon) M)$. For simplicity, we use the notation $f_{\mathcal{C}_2}(c)$ to denote the mapping from a codeword $c \in \mathcal{C}_1$ to a corresponding codeword in $\mathcal{C}_1 \circ \mathcal{C}_2$.

Let $G_1$ be a $(K, \rho)$-disperser with $N$ left vertices and $M$ right vertices, with parameters given by Lemma \ref{lemma_disp} with $K = \Theta(\varepsilon N)$. Let $G_2$ be a $(K', \delta \delta')$-multi-set disperser with $N'$ left vertices and $M'$ right vertices, with parameters given by Lemma \ref{lemma_rand_mul} with $K' = \Theta(\varepsilon N')$. Since $M = \Theta(\varepsilon \log(1/\varepsilon) N)$ for $G_1$ and $M' = \Theta(\varepsilon \log(1/\varepsilon) N')$ for $G_2$, we know that $N_2 = \Theta(N_1 / \varepsilon)$. We set $t = \lfloor \frac{N_2}{N_1} \rfloor = \Theta(\log(1/\varepsilon))$.

Then, we use these four objects to construct a new code, denoted as $G_1(\mathcal{C}_1) \star G_2(\mathcal{C}_1 \circ \mathcal{C}_2)$, of length $N$ over the alphabet $\mathbb{F}_q^D \times (\mathbb{F}_2^D)^t$, with the same message space as that of $\mathcal{C}_1$. For a codeword $c \in \mathcal{C}_1$, the decoded codewords, denoted as $G_1(c)\star G_2(f_{\mathcal{C}_2}(c)) \in G_1(\mathcal{C}_1) \star G_2(\mathcal{C}_1 \circ \mathcal{C}_2)$ is given by:
$$
\left( \left[ \begin{array}{c}
G_1(c)_1 \\
G_2(f_{\mathcal{C}_2}(c))_1 \\
G_2(f_{\mathcal{C}_2}(c))_2 \\
\vdots \\
G_2(f_{\mathcal{C}_2}(c))_{t}
\end{array} \right], \left[ \begin{array}{c}
G_1(c)_2 \\
G_2(f_{\mathcal{C}_2}(c))_{t+1} \\
G_2(f_{\mathcal{C}_2}(c))_{t+2} \\
\vdots \\
G_2(f_{\mathcal{C}_2}(c))_{2t}
\end{array} \right], \ldots, \left[ \begin{array}{c}
G_1(c)_{N} \\
G_2(f_{\mathcal{C}_2}(c))_{t(N-1)+1} \\
G_2(f_{\mathcal{C}_2}(c))_{t(N-1)+2} \\
\vdots \\
G_2(f_{\mathcal{C}_2}(c))_{tN}
\end{array} \right] \right) \in (\mathbb{F}_q^D \times (\mathbb{F}_2^D)^t)^N.
$$

We lose some values of $G_2(f_{\mathcal{C}_2}(c))$ when performing the folding operation on it, but since this proportion is at most $\Theta(1/\varepsilon)$, the impact on the code rate is at most a constant factor. It can be computed that the rate of this code is still $\Omega(\varepsilon)$ and the alphabet size is at most $(t+1)\poly(1/\varepsilon)$, which is still $\poly(1/\varepsilon)$ since $t=\Theta(\log(1/\varepsilon))$

Similarly, it can be observed that this code is still $(1 - O(\varepsilon \log(1/\varepsilon)), L)$-list decodable for list size $L = O(1/\varepsilon)$ given that $G_2(\mathcal{C}_1\circ \mathcal{C}_2)$ is $(1 - O(\varepsilon \log(1/\varepsilon)), L)$-list decodable. For the decoding algorithm, we use the list decoding algorithm for $G_1(\mathcal{C}_1)$ to get a list of codewords $\mathcal{S}^{\prime}$ from $G_1(\mathcal{C}_1)$ with a size of at most $O_q(M)$. Then, for each $G_1(c) \in \mathcal{S}^{\prime}$, we compute $G_1(c) \star G_2(f_{\mathcal{C}_2}(c))$ and compute the distance from $y$ to check whether we should keep it. By screening all candidates in this list, we can downsize the list to at most $O(1/\varepsilon)$. The algorithm is provided in Algorithm \ref{algorithm_3}.

\begin{algorithm}[h]
\label{algorithm_3}
    \caption{List Decoder up to Near-Optimal List Size}
    \SetAlgoLined
    \KwIn{Received word $y: L \mapsto \mathbb{F}_q^D \times (\mathbb{F}_2^D)^t$}
    \KwOut{A set $\mathcal{S}$ containing all codewords $G_1(c)\star G_2(f_{\mathcal{C}_2}(c)) \in G_1(\mathcal{C}_1) \star G_2(\mathcal{C}_1 \circ \mathcal{C}_2)$ with $d(G_1(c)\star G_2(f_{\mathcal{C}_2}(c)), y) \leq (1 - \gamma)N$ for some $\gamma=O(\varepsilon \log(1/\varepsilon))$}
    Set $y^{\prime}: L \mapsto \mathbb{F}_q^D$ as the words achieved by restricting the value of $y$ to $\mathbb{F}_q^D$\; 
    Use Algorithm \ref{algorithm_2} to decode $y^{\prime}$ in $G_1(\mathcal{C}_1)$ and get a list $\mathcal{S}^{\prime}$ with size $O_q(M)$\;

    \ForEach{$G_1(z^{\prime}) \in \mathcal{S}^{\prime}$}{ 
    $z = G_1(z^{\prime})\star G_2(f_{\mathcal{C}_2}(z^{\prime}))$\;
    \If{$d(y, z) \leq (1 - \gamma)N$}{
        $\mathcal{S} = \mathcal{S} \cup z$\;
    }}
    \Return $\mathcal{S}$\;
\end{algorithm}

By slightly adjusting the parameters to achieve a list decoding radius of $1 - \varepsilon$, we obtain the following result.

\begin{corollary}
    For any $\varepsilon > 0$, there exists a family of codes over an alphabet of size $\poly(1/\varepsilon)$, which has rate $\Omega\left(\frac{\varepsilon}{\log^2(1/\varepsilon)}\right)$ such that a code with block length $n$ in the family can be list decoded from up to $1-\varepsilon$ fraction of errors with list size $L = O\left(\frac{\log^2(1/\varepsilon)}{\varepsilon}\right)$ in time $\poly_{\varepsilon}(n) $, and can be constructed in probabalistic time $\poly(n,\log(1/\varepsilon))$ with success probability at least $1-\exp(-\log(1/\varepsilon)\cdot n)$.
\end{corollary}
\section{Open Problems and Future Directions}\label{sec_open}

\subsection{Reducing the Degree of Explicit Dispersers}

One of our open problems is to reduce the degree in explicit dispersers. In our code construction, we use explicit dispersers given by \cite{capalbo2002randomness}. While achieving an entropy loss of $O_{\delta}(1)$, the degree of the disperser is $\poly_{\delta}(N/K)$, which is not optimal. This results in the concatenated code having an alphabet size of $\quasipoly(1/\varepsilon)$. Optimally, but not explicitly, the degree of the disperser can be reduced to $O_{\delta}(\log(N/K))$, while maintaining a constant entropy loss for a fixed $\delta$. However, how to derandomize this construction remains unknown.

\subsection{Improving Results in Plurality Analysis}

In the analysis of list-decodable codes with optimal list sizes, we adopt plurality analysis, but some results still need improvement or derandomization. In the analysis of Theorem \ref{theorem_last}, we question whether the rate can be improved from $\Omega(\frac{\varepsilon}{ \log^2(1/\varepsilon)})$ to $\Omega(\varepsilon)$ and whether the list size can be improved from $O(\frac{ \log^2(1/\varepsilon)}{\varepsilon})$ to $O(\frac{1}{\varepsilon})$. We conjecture that this can be achieved by improving the analysis method without changing the code structure. Our most interesting question, of course, is to completely derandomize the multi-set disperser with optimal parameters as stated in the Lemma \ref{lemma_rand_mul}, which would directly yield an explicit list-decodable code with near optimal list size and near optimal rate.

The plurality analysis does not explicitly provide a list decoding algorithm. Instead, we combine such codes with codes provided by Theorem \ref{theorem_list1} and essentially follow the previous decoding method by employing the folding operation. An interesting question is whether we can directly start from plurality analysis to provide an effective decoding algorithm.

\subsection{Applications to Other Regimes}

We use graph concatenation to construct uniquely decodable codes and list decodable codes in the high-noise regimes. One interesting future direction is to provide codes in other regimes or to construct codes with different properties. An intriguing possibility is to develop codes that achieve the general Singleton bound: $ R = 1 - \delta - \varepsilon $ or the general list decoding capacity: $ R = 1 - \rho - \varepsilon $. Additionally, we question whether this method can be applied to codes over smaller alphabets, such as binary codes.

\section*{Acknowledgment}
We thank the anonymous reviewers for their useful comments that helped improve the presentation of this paper. We also thank Yeyuan Chen and Zihan Zhang for pointing out an inaccuracy in the introduction of an earlier version.

\bibliographystyle{alpha}
\bibliography{references}

\end{document}